\documentclass[11pt]{article}
\usepackage{amssymb}

\usepackage{graphicx}
\usepackage{amsmath}
\usepackage{makeidx}
\usepackage{indentfirst}

\newcounter{resultnum}[section]
\newtheorem{conclusion}{Conclusion}[section]

\newcounter{conclusionnum}[section]

\newcounter{conditionnum}[section]

\newcounter{conjecturenum}[section]

\newcounter{examplenum}[section]

\newcounter{exercisenum}[section]
\newtheorem{lemma}{Lemma}[section]

\newcounter{lemmanum}[section]

\newcounter{notationnum}[section]
\newtheorem{theorem}{Theorem}[section]

\newcounter{theoremnum}[section]
\newtheorem{definition}{Definition}[section]

\newcounter{definitionnum}[section]
\newtheorem{corollary}{Corollary}[section]

\newcounter{corollarynum}[section]
\newtheorem{remark}{Remark}[section]

\newcounter{remarknum}[section]
\newtheorem{proposition}{Proposition}[section]

\newcounter{propositionnum}[section]

\newcounter{acknowledgementnum}[section]

\newcounter{algorithmnum}[section]

\newcounter{axiomnum}[section]

\newcounter{casenum}[section]
\newtheorem{claim}{Claim}[section]

\newcounter{claimnum}[section]

\newcounter{summarynum}[section]

\newcounter{problemnum}[section]
\newenvironment{proof}[1][]{\textbf{Proof.} }{}

\begin{document}

\title{Dynamical Equations and Lagrange--Ricci Flow Evolution on Prolongation Lie  Algebroids}
\date{March 02, 2018}

\author{
${}$ \\
Lauren\c{t}iu Bubuianu\\
{\small \textit{TVR Ia\c{s}i, \ 33 Lasc\v{a}r Catargi street, 700107 Ia\c{s}i, Romania }} \\
{\small and \textit{University Apollonia, 2 Muzicii street, Ia\c{s}i, Romania }} \\
{\small \textit{email: laurentiu.bubuianu@tvr.ro }}\\
${}$ \\
{Sergiu I. Vacaru}\footnote{{\it Address for correspondence:\ } 67 Lloyd Street South, Manchester, M14 7LF, the UK }\\
%\vspace{.1 in}
{\small \textit{Physics Department, California State University at Fresno,  CA 93740, USA}} \\
{\small and } \ {\small \textit{ Project IDEI, University "Al. I. Cuza" Ia\c si, Romania}}\\
{\small \textit{email: sergiu.vacaru@gmail.com ; sergiuvacaru@mail.fresnostate.edu}}
 }
\maketitle

\begin{abstract}
The approach to nonholonomic Ricci flows and geometric evolution of regular Lagrange systems [S. Vacaru: J. Math. Phys. \textbf{49} (2008) 043504 \& Rep. Math. Phys. \textbf{63} (2009) 95] is extended to include geometric mechanics and gravity models on Lie algebroids. We prove that such evolution scenarios of geometric mechanics and analogous gravity can be modeled as gradient flows characterized by generalized Perelman functionals if an equivalent geometrization of Lagrange mechanics [J. Kern, Arch. Math. (Basel) \textbf{25} (1974) 438] is considered. The R. Hamilton equations on Lie algebroids describing Lagrange-Ricci flows are derived. Finally, we show that geometric evolution models on Lie algebroids are described by effective thermodynamical values derived from  statistical functionals on prolongation Lie algebroids.
\vskip0.1cm

 \vskip3pt

%MSC2010:\ 53C44, 17B66, 37J60, 53D17, 70G45, 70S05
\end{abstract}

%\tableofcontents

\section{Introduction}

The Ricci flow theory \cite{ham1,ham2} became attractive for research in
mathematics and physics after G. Perelman successfully carried out his
program \cite{gper1,gper2,gper3} which resulted in proofs of Thurston and
Poincar\'{e} conjectures, see reviews  \cite%
{caozhu,kleiner,rbook}. The profound impact of such results on understanding
the topology and geometric structure of curved spacetime and fundamental
properties of classical and quantum interactions was used as a motivation to
study the geometric evolution of regular Lagrange systems \cite%
{vricci1,vricci2} on tangent bundles and nonholonomic (pseudo) Riemannian
and Einstein manifolds. We developed Ricci flow theories for classical and
quantum  solutions of Einstein equations, generalizations to
noncommutative, Finsler, diffusion, fractional spaces etc, see \cite%
{vncricci,vfinsler} and references therein.

Effective Lagrange and Hamilton models, Lie algebroid and almost K\"{a}hler
and Dirac structures are considered, for instance, in quantum gravity and
modified gravity theories when Ricci flows on parameters are derived from a
renormalization procedure with running/evolution of physical parameters etc \cite%
{cortes,vrflg,vqgr2}. One of the important tasks in modern geometry and
physics is to elaborate and analyze flow evolution of more complex
geometries and physical systems with nontrivial topology, generalized
symmetries, nonholonomic constraints etc. So, it is not an academic exercise and "pure" geometric interest to perform generalizations of the Ricci flow
theory derived for Lagrangians/Hamiltonians on Lie algebroids. Fundamental properties of spacetime topology seem to be related to a series of important questions on dimensions of real mechanical systems and physical interactions, analogous gravity modelling, renormalizability of certain quantum theories, possible modifications of gravity derived from modern cosmology observations etc. We need rigorous studies on evolution of theories with rich geometric structure, generalized and deformed symmetries and symplectic structure and nonholonomic constraints.

Specifically, the goal of this paper is to elaborate a model of geometric
evolution of Lagrange mechanics and analogous gravity theory on Lie algebroids using certain
constructions proposed and developed in Refs. \cite%
{kern,matsumoto,valg1,valg2}. The key idea considered in our works is that
physical theories can be encoded into the geometry of generalized nonholonomic spaces
(defined by corresponding classes of non--integrable constraints on
fundament dynamical and evolution equations) via "standard", or analogous,
geometric objects like metrics, (almost) symplectic forms, nonlinear and
linear connections, related curvatures and torsions, and their geometric
flows evolution. A subclass of evolution scenarios are uniquely determined
following geometric principles for entropy type functionals derived for
families of generating Lagrange functions $L(\mathbf{x,y},\chi )$.\footnote{%
We treat $(\mathbf{x,y)}$ as some generalized coordinates (for instance, on
a tangent bundle $TM,$ or a Lie algebroid $E$ over a manifold $M$) and $\chi$
is a real evolution parameter. In certain ''dual'' and, in some sense, more
general approaches, we can consider families of Hamiltonians $H(\mathbf{x,p}%
,\chi )$, (almost) symplectic and/or Poisson structures with associated
co--tangent bundle $T^{\ast }M$ etc. Here we also note that we use left low/up indices as labels for some geometric objects
and/or spaces.} Hopefully, such assumptions on geometric evolution mechanics
allows us to formulate an alternative and very different approach to and
provide us new possibilities to explore properties of Lagrange systems using
methods in geometric analysis.

Theories of Lagrange and Hamilton systems on Lie algebroids (and various
discrete analogs on Lie groupoids, Poisson structures and algebroids etc)
were proposed \cite{weinstein,libermann} and actively developed during last
ten years, see original contributions and reviews of results in Refs. \cite%
{cortes,martinez1,martinez2,dleon1}. On Lie algebroid gravity and gauge
interactions models, we cite \cite{strobl1,valg1,valg2} and references
therein. The inclusive nature of Lie algebroid formalism allows us to
describe very different situations in mechanics and physics such as
Lagrangian systems with symmetry and nonholonomic constraints, theories with
semidirect products and/or evolving on Lie algebras and generalizations. It
is possible in such cases to derive some Lagrange/Euler --Poincar\'{e}, or
Euler--Lagrange equations and geometrize such systems as generalized Poisson
geometries etc. New tools have been introduced and new understanding have
been provided, for instance, by the multi--symplectic formalism and
Poisson--Nijenhuis Lie algebroid theory etc.

Nevertheless, we have to consider additional and alternative constructions
for above mentioned algebroid models of geometric mechanics and
classical/quantum field theories if we wont to study the Ricci flow
evolution of systems and spaces with ''rich'' geometric and physical
structure keeping certain analogy with the Hamilton--Perelman theory. It is
not clear how the standard formalism elaborated for Ricci flows of (semi)
Riemann and (almost) K\"{a}hler geometries can be extended to describe
directly flow evolution of models of Lie algebroid mechanics developed in
Refs. \cite{weinstein,libermann,martinez1,martinez2,dleon1}.

Our proposal is to use J. Kern's \cite{kern} constructions on Lagrange
spaces (the term is due to that article developing in a ''nonhomogenous''
manner the M. Matsumoto results on Finsler connections \cite{matsumoto}, see
references therein; further developments and applications, for instance, in
modern classical and quantum gravity \cite{vncricci,vrflg,vqgr2}. In such an
approach, the nondegenerate Hessian of a regular Lagrangian can be treated
as a metric structure for fibers on $TM$ which can be extended on total
space using the so--called Sasaki lifts \cite{yano}. It is involved also a
corresponding semi--spray structure inducing a canonical nonlinear
connection (in brief, N--connection; the global definition is due to \cite%
{ehresmann}, see historical remarks and applications in modern mechanics and
gravity in \cite{vrflg}). For such geometric data, a model of
Lagrange--Ricci flow theory \cite{vricci1,vricci2} can be formulated in
N--adapted form, via corresponding generalizations of Perelman's
functionals, on tangent bundles and/or nonholonomic\footnote{%
equivalently, anholonomic and/or non--integrable, which refer to certain
classes of distributions on a geometric/physical space} (semi) Riemannian
manifolds.

The paper is organized as follows. In section \ref{s2}, we survey the
geometry of Lie algebroids and prolongations and geometrization of Lagrange
mechanics on such spaces following approach from \cite%
{cortes,martinez1,martinez2,dleon1}. We summarize necessary tools from the
geometry of N--connections on prolongation Lie algebroids in section \ref{s3}%
. The constructions are performed in metric compatible form which allows us
to formulate an analogous N--adapted gravity model on Lie algebroids. An
alternative geometrization of regular Lagrange mechanics and analogous
modeling of gravity following Kern--Matsumoto ideas extended for
prolongation Lie algebroids is provided. Section \ref{s4} is devoted to Main
Theorems for Lagrange--Ricci flows on prolongation Lie algebroids. \vskip5pt

\section{ Lagrange Mechanics and Lie Algebroids}

\label{s2}We outline basic concepts and definitions for Lie algebroids and
geometric mechanics with regular Lagrangians on prolongations of Lie
algebroids on bundle maps, see  \cite{cortes,martinez1,dleon1} and
references therein.

\subsection{Linear connections and metrics on Lie algebroids}

A \textbf{Lie algebroid }$\mathcal{E}=(E,\left\lfloor \cdot ,\cdot
\right\rfloor ,\rho )$ over a manifold $M$ is a triple defined by 1) a real
\textit{vector bundle} $\tau :E\rightarrow M$ together with 2) a \textit{Lie bracket} $\left\lfloor \cdot ,\cdot \right\rfloor $ on the spaces of global
sections $Sec(\tau )$ \ of map $\tau $ and 3) the \textit{anchor} map
$\rho :E\rightarrow TM$ defined as a bundle map over identity and constructed such
that the homorphism $\rho :Sec(\tau )\rightarrow \mathcal{X}(M)$ of $%
C^{\infty }(M)$--modules $\mathcal{X}$ induced this map satisfies the
condition%
\begin{equation*}
\left\lfloor X,fY\right\rfloor =f\left\lfloor X,Y\right\rfloor +\rho
(X)(f)Y,~\forall X,Y\in Sec(\tau )\mbox{ and }f\in C^{\infty }(M).
\end{equation*}%
For a Lie algebroid, the anchor map $\rho $ is equivalent to a homomporphysm between the Lie algebras $\left( Sec(\tau ),\left\lfloor \cdot ,\cdot \right\rfloor \right) $ and $\left( \mathcal{X}(M),\left\lfloor \cdot ,\cdot \right\rfloor \right) .$

In local form, the properties of a Lie algebroid $\mathcal{E}$ are
determined by the local functions $\rho _{\alpha }^{i}(x^{k})$ and $%
C_{\alpha \beta }^{\gamma }(x^{k})$ on $M,$ where $x=\{x^{k}\}$ are local
coordinates on a chart $U\subset M,$ with $\rho (e_{\alpha })=\rho _{\alpha
}^{i}(x)\partial _{i}$ and $\left\lfloor e_{\alpha },e_{\beta }\right\rfloor
=C_{\alpha \beta }^{\gamma }(x)e_{\gamma }$, satisfying the following
equations%
\begin{equation*}
\rho _{\alpha }^{i}\partial _{i}\rho _{\beta }^{j}-\rho _{\beta
}^{i}\partial _{i}\rho _{\alpha }^{j}=\rho _{\gamma }^{j}C_{\alpha \beta
}^{\gamma }\mbox{ \ and \ }\sum\limits_{\mbox{
cyclic \ }(\alpha ,\beta ,\gamma )}\left( \rho _{\alpha }^{i}\partial
_{i}C_{\beta \gamma }^{\nu }+C_{\beta \gamma }^{\mu }C_{\alpha \mu }^{\nu
}\right) =0.
\end{equation*}

A linear connection $D$ on $\mathcal{E}$ is defined as a $\mathbb{R}$%
--bilinear map $D:Sec(E)\times Sec(E)\rightarrow Sec(E)$ such that $\forall
~f\in C^{\infty }(M)$ and $\forall ~X,Y\in Sec(E)$ this covariant derivative
operator satisfies the conditions
\begin{equation}
D_{fX}Y=fD_{X}Y\mbox{\ and \ }D_{X}(fY)=\rho (X)(f)Y+fD_{X}Y.  \label{cond1}
\end{equation}%
Locally, $D$ is given by its coefficients $\Gamma _{~\beta \gamma }^{\gamma}
$ when $D_{X}Y=X^{\alpha }(\rho _{\alpha }^{i}\partial _{i}Y^{\gamma
}+\Gamma _{~\alpha \beta }^{\gamma }Y^{\beta })e_{\gamma }$, where $%
X=X^{\alpha }e_{\alpha }$ and $Y=Y^{\alpha }e_{\alpha }$ for a local basis $%
\{e_{\alpha }\}\in Sec(E).$ A curve $a:I\rightarrow E$ given by a function $%
a(\tau )=$ $a^{\alpha }(\tau )e_{\alpha }$ on a real parameter $\tau ,$ is
said to be an auto--parallel of $D$ if $D_{a}a=0.$

The exterior differential on $\mathcal{E}$ can be defined in standard form
using the operator $d$ on $\mathcal{E},$ when $d:Sec(\bigwedge\nolimits^{k}%
\tau ^{\ast })\rightarrow Sec(\bigwedge\nolimits^{k+1}\tau ^{\ast
}),d^{2}=0, $ where $\bigwedge $ is the antisymmetric product operator, see
details in Refs. \cite{higgins, martinez1,dleon1,cortes}. The local
contributions of a N--connection can be seen from such formulas for a smooth
formula $f:M\rightarrow \mathbb{R},df(X)=\rho (X)f,$ for $X\in Sec(\tau ),$
when
\begin{equation*}
dx^{i}=\rho _{\alpha }^{i}e^{\alpha } \mbox{ and }de^{\gamma }=-\frac{1}{2}%
C_{\alpha \beta }^{\gamma }e^{\alpha }\wedge e^{\beta }.
\end{equation*}
With respect to any section $X,$ we can define the Lie derivative%
\begin{equation}
\mathcal{L}_{X}=i_{X}\circ d+d\circ i_{X}:~Sec(\bigwedge\nolimits^{k}\tau
^{\ast })\rightarrow Sec(\bigwedge\nolimits^{k}\tau ^{\ast }),
\label{lderiv}
\end{equation}%
using the cohomology operator $d$ and its inverse $i_{X},$ see details in %
\cite{mackenzie,dleon1,cortes}.

A metric $\varpi $ on $\mathcal{E}$ \ is defined as a map
\begin{equation}
\varpi :E\times _{M}E\rightarrow \mathbb{R}.  \label{algmetric}
\end{equation}%
Locally, $\varpi =\varpi _{\alpha \beta }(x)e^{\alpha }\otimes e^{\beta }.$
We shall use also the inverse matrix/ metric, $\varpi ^{\alpha \beta }(x).$

There is a ''preferred'' linear connection $~^{\varpi }\nabla $ on $\mathcal{%
E}$ \ (the analog of the Levi--Civita connection in Riemannian geometry)
completely defined by a metric $\varpi .$ This connection is uniquely
determined following two conditions:%
\begin{eqnarray*}
~^{\varpi }T(X,Y)=\ ^{\varpi }\nabla _{X}Y-~^{\varpi }\nabla
_{Y}X+\left\lfloor X,Y\right\rfloor &=&0,\mbox{ i.e. zero torsion}; \\
\varpi (~^{\varpi }\nabla _{X}Y,Z)+\varpi (Y,~^{\varpi }\nabla _{X}Z)-\rho
(X)\left( \varpi (Y,Z)\right) &=&0,\mbox{i.e. metricity},
\end{eqnarray*}%
which result in formula%
\begin{eqnarray*}
2\varpi (~^{\varpi }\nabla _{X}Y,Z) &=&\rho (X)\left( \varpi (Y,Z)\right)
+\rho (Y)\left( \varpi (X,Z)\right) -\rho (Z)\left( \varpi (X,Y)\right) \\
&&-\varpi (~X,\left\lfloor Y,Z\right\rfloor )+\varpi (~Y,\left\lfloor
Z,X\right\rfloor )-\varpi (~Z,\left\lfloor Y,X\right\rfloor ),
\end{eqnarray*}%
$\forall X,Y,Z\in Sec(E).$ The curvature of $~^{\varpi }\nabla $ on $%
\mathcal{E}$ (the analog of the Riemannian tensor on standard manifolds) is
defined in standard from
\begin{equation*}
~^{\varpi }R(X,Y)Z=\left( \ ^{\varpi }\nabla _{X}~^{\varpi }\nabla
_{Y}-~^{\varpi }\nabla _{Y}\ ^{\varpi }\nabla _{X}~-~^{\varpi }\nabla
_{\left\lfloor X,Y\right\rfloor }\right) Z.
\end{equation*}

Introducing in above formulas $X=e_{\alpha },Y=e_{\beta },Z=e_{\gamma },$
for {\small
\begin{eqnarray}
\ ^{\varpi }\nabla _{e_{\alpha }}e_{\beta } &:= &\ ^{\varpi }\nabla _{\alpha
}e_{\beta }=\Gamma _{~\beta \alpha }^{\gamma }e_{\gamma },  \label{lccon} \\
\Gamma _{~\beta \gamma }^{\alpha } &=&\frac{1}{2}\varpi ^{\alpha \varphi
}(\rho _{\gamma }^{i}\partial _{i}\varpi _{\beta \varphi }+\rho _{\beta
}^{i}\partial _{i}\varpi _{\gamma \varphi }-\rho _{\varphi }^{i}\partial
_{i}\varpi _{\beta \gamma }  \notag \\
&&+C_{\varphi \gamma }^{\tau }\varpi _{\tau \beta }+C_{\varphi \beta }^{\tau
}\varpi _{\tau \gamma }-C_{\beta \gamma }^{\tau }\varpi _{\tau \varphi }),
\notag
\end{eqnarray}
} we compute the coefficients of torsion and curvature of the Levi--Civita
connection, respectively,
\begin{eqnarray*}
\ ^{\varpi }T_{~\beta \alpha }^{\gamma } &=&\Gamma _{~\beta \alpha }^{\gamma
}-\Gamma _{~\alpha \beta }^{\gamma }+C_{\alpha \beta }^{\gamma }=0%
\mbox{ and
} \\
\ ^{\varpi }R_{~\beta \gamma \delta }^{\alpha } &=&\rho _{\delta
}^{i}\partial _{i}\Gamma _{~\beta \gamma }^{\alpha }-\rho _{\gamma
}^{i}\partial _{i}\Gamma _{~\beta \delta }^{\alpha }+\Gamma _{~\beta \gamma
}^{\varphi }\Gamma _{~\varphi \delta }^{\alpha }-\Gamma _{~\beta \delta
}^{\varphi }\Gamma _{~\varphi \gamma }^{\alpha }+\Gamma _{~\beta \varphi
}^{\alpha }C_{\gamma \delta }^{\varphi }.
\end{eqnarray*}%
In standard form, we define the Ricci tensor contracting respective indices,
$\ ^{\varpi }Ric=\{~^{\varpi }R_{~\beta \gamma }:= \ ^{\varpi }R_{~\beta
\gamma \alpha }^{\alpha }\}$, and the scalar curvature, $~^{\varpi
}R:=\varpi ^{\alpha \beta }~^{\varpi }R_{~\alpha \beta }.$ Such formulas are
very similar to those for (pseudo) Riemannian geometry formulated in
nonholonomic bases satisfying anholonomy relations for some nontrivial
coefficients $C_{\gamma \delta }^{\varphi }.$ For the case of Lie
algebroids, the fundamental geometric objects on the space $Sec(E).$ The
above formulas on metrics, connections and Ricci tensors can be used for
elaborating a Ricci Lie algebroid evolution theory. Nevertheless, there are
necessary a number of additional assumptions and constructions in order to
include in such a scheme models of Lagrange mechanics and classical and
quantum field theories.

\subsection{The prolongation of Lie algebroids and Lagrange mechanics}

In Refs. \cite{martinez1,dleon1,cortes}, a geometric formalism for Lagrange
mechanics on Lie algebroids was developed using the concept of prolongation
of a Lie algebroid over a fibration (in brief, prolongation Lie algebroid).
Let us briefly outline some basic constructions.

Consider a Lie algebroid $\mathcal{E}=(E,\left\lfloor \cdot ,\cdot
\right\rfloor ,\rho )$ and a fibration $\pi :P\rightarrow M$ \ both defined
over the same manifold $M.$ We denote local coordinates in the form $%
(x^{i},u^{A})\in P$ \ write $\{e_{\alpha }\}$ for a local basis of sections
of $E.$ For our purposes, we can consider that $P=E$. The anchor map $\rho
:E\rightarrow TM$ and the tangent map $T\pi :TP\rightarrow TM,$ can be used
to construct a subset $\mathcal{T}_{s}^{E}P:=\{(b,v)\in E_{x}\times
T_{x}P;\rho (b)=T_{p}\pi (v);p\in P_{x},\pi (p)=x\in M\}$. Globalizing the
construction, we obtain another Lie algebroid, $\mathcal{T}%
^{E}P:=\bigcup_{s\in S}\mathcal{T}_{s}^{E}P,$ which is called the
prolongation of $E$ over $\pi .$ Equivalently, $\mathcal{T}^{E}P$ is called
the $E$--tangent bundle to $\pi ,$ which is also a vector bundle over $P$,
with projection $\tau _{P}^{E}$ just onto the first factor, $\tau
_{P}^{E}(b,v)=b.$ The elements of $\mathcal{T}^{E}P$ are written $(p,b,v).$
There are also used brief denotations $(p,b,v)\in \mathcal{T}%
_{p}^{E}P\rightarrow $ $(b,v)\in \mathcal{T}^{E}P$ if not ambiguities. The
anchor $\rho ^{\pi }:\mathcal{T}^{E}P\rightarrow \mathcal{T}P$ is given by
maps $\rho ^{\pi }$ $(p,b,v)=v,$ i.e. projection onto the third factor.

It is possible to define also the projection onto the second factor (i.e. a
morphism of Lie algebroids over $\pi ),$ $\mathcal{T}\pi :\mathcal{T}%
^{E}P\rightarrow E,$ when $\mathcal{T}\pi (p,b,v)=b.$ For instance, an
element $(p_{1},b_{1},v_{1})\in $ $\mathcal{T}^{E}P$ is vertical if $%
\mathcal{T}\pi $ $(p_{1},b_{1},v_{1})=b_{1}=0,$ i.e. such elements are of
type $(p,0,v)$ when $v$ is a $\pi $--vertical vector (tangent to $P$ at
point $p).$

Locally, any element $\overline{z}=(p,b,v)\in $ $\mathcal{T}^{E}P,$ when $%
b=z^{\alpha }e_{\alpha }$ and $v=\rho _{\alpha }^{i}z^{\alpha }\partial
_{i}+v^{A}\partial _{A},$ for $\partial /\partial u^{A},$ can be decomposed $%
\overline{z}=z^{\alpha }\mathcal{X}_{\alpha }+v^{A}\mathcal{V}_{A}$, where $%
\left( \mathcal{X}_{\alpha },\mathcal{V}_{A}\right) ,$ with vertical $%
\mathcal{V}_{A},$ define a local basis of sections of $\mathcal{T}^{E}P.$ In
explicit form, such bases can be parametrized in the form $\mathcal{X}%
_{\alpha }=\mathcal{X}_{\alpha }(p)=\left( e_{\alpha }(\pi (p)),\rho
_{\alpha }^{i}\partial _{i\mid p}\right)$ and $\mathcal{V}_{A}=\left(
0,\partial _{A\mid p}\right)$ where partial derivatives are taken in a point
$p\in S_{x}.$

The Lie algebroid structure of $\mathcal{T}^{E}P$ is stated by the anchor
map $\rho ^{\pi }(Z)=\rho _{\alpha }^{i}Z^{\alpha }\partial
_{i}+V^{A}\partial _{A}$ acting on sections $Z$ with associated
decompositions of type $\overline{z}$ and by the Lie brackets $\left\lfloor
\mathcal{X}_{\alpha },\mathcal{X}_{\beta }\right\rfloor ^{\pi }=C_{\alpha
\beta }^{\gamma }\mathcal{X}_{\gamma },~\left\lfloor \mathcal{X}_{\alpha },%
\mathcal{V}_{B}\right\rfloor ^{\pi }=0,~~\left\lfloor \mathcal{V}_{A},%
\mathcal{V}_{B}\right\rfloor ^{\pi }=0.$ Using dual bases $\left( \mathcal{X}%
^{\alpha },\mathcal{V}^{B}\right) $, we can perform an exterior differential
calculus following formulas
\begin{equation}
dx^{i}=\rho _{\alpha }^{i}\mathcal{X}^{\alpha },\mbox{ for  }d\mathcal{X}%
^{\gamma }=-\frac{1}{2}C_{\alpha \beta }^{\gamma }\mathcal{X}^{\alpha
}\wedge \mathcal{X}^{\beta },\mbox{ and \ }du^{A}=\mathcal{V}^{A},%
\mbox{ \
for \ }d\mathcal{V}^{A}=0.  \label{extercalc}
\end{equation}%
For instance, if we take a (function, or Lagrangian) $L(x^{i},u^{\alpha })$
on $E$ we can compute%
\begin{equation*}
d^{E}L=\rho _{\alpha }^{i}(\partial _{i}L)\mathcal{X}^{\alpha }+(\partial
_{\alpha }L)\mathcal{V}^{\alpha },
\end{equation*}%
where $d^{E}x^{i}=\rho _{\alpha }^{i}\mathcal{X}^{\alpha }$ and $%
d^{E}u^{\alpha }=\mathcal{V}^{\alpha }.$ We shall write the absolute
differential, for instance, of $L,$ in the form $d^{E}L$ for $%
dx^{i}\rightarrow d^{E}x^{i}$ and $du^{\alpha }\rightarrow d^{E}u^{\alpha }$
if $P=E$ and $\mathcal{V}^{A}\rightarrow \mathcal{V}^{\alpha }.$

Let us consider $P=E$ for $\mathcal{T}^{E}P$ when the prolongation Lie
algebroid is for a bundle projection $\tau :E\rightarrow M$. We can formulate
a mechanical model for a Lagrangian function $L\in C^{\infty }(E)$ and chose
a vertical endomorphism $S:\mathcal{T}^{E}E\rightarrow \mathcal{T}^{E}E$ of
type $S(a,b,v)=\xi ^{V}(a,b)=(a,0,b_{a}^{V}),$ where $b_{a}^{V}$ is the
vector tangent to the curve $a+\tau b$ when the parameter $\tau =0.$ The
vertical lift $\xi ^{V}$ allows us to define a map $\xi ^{V}:\tau ^{\ast
}E\rightarrow \mathcal{T}^{E}E$ and the Liouville dilaton vector field $%
\bigtriangleup (a)=\xi ^{V}(a,a)=(a,0,b_{a}^{V}).$

A model of Lie algebroid mechanics for a Lagrangian $L$ can be geometrized
on $\mathcal{T}^{E}E$ in terms of three geometric objects,%
\begin{eqnarray}
\mbox{ the Cartan 1-section:  }\theta _{L}:= &&S^{\ast }(dL)\in Sec((%
\mathcal{T}^{E}E)^{\ast });  \label{cartvar} \\
\mbox{ the Cartan 2-section:  }\omega _{L}&:= &-d\theta _{L}\in Sec(\wedge
^{2}(\mathcal{T}^{E}E)^{\ast });  \notag \\
\mbox{ the Lagrangian energy : } E_{L}&:= &\mathcal{L}_{\bigtriangleup}L
-L\in C^{\infty }(E),  \notag
\end{eqnarray}%
where the Lie derivative (\ref{lderiv}) is considered in the last formula.
Using these variables, the dynamical equations \ derived for $L$ can be
geometrized as
\begin{equation}
i_{SX}\omega _{L}=-S^{\ast }(i_{X}\omega _{L})\mbox{ \ and \ }%
i_{\bigtriangleup }\omega _{L}=-S^{\ast }(dE_{L}),\forall X\in Sec(\mathcal{T%
}^{E}E).  \label{geomeq}
\end{equation}%
Such geometric equations define equivalently a regular Lagrange mechanics if
$\omega _{L}$ is regular at every point as a bi--linear form, i.e. it is a
symplectic section. For configurations with regular $L$, and $\omega _{L},$
there exists a unique solution $\Gamma _{L}$ and a form $\Omega _{L}$
satisfying the condition $i_{\Gamma _{L}}\Omega _{L}=dE_{L}$. From equations
(\ref{geomeq}), we obtain $i_{S\Gamma _{L}}\omega _{L}=i_{\bigtriangleup
}\omega _{L}.$ This states that $S(\Gamma _{L})=\bigtriangleup $
(equivalently, $\mathcal{T}\tau (\Gamma _{L}(a))=a,\forall a\in E)$ which
constraints $\Gamma _{L}$ to be a SODE (second order differential equation)
section, or semispray. Taking $\Omega _{L}=\omega _{L},$ for $P=E,$ we can
write the last equation as a symplectic equation%
\begin{equation}
i_{\Gamma _{L}}\omega _{L}=d^{E}E_{L},  \label{geomeq1}
\end{equation}%
for $\Gamma _{L}\in Sec(\mathcal{T}^{E}E).$

The above geometric objects (\ref{cartvar}) and equations (\ref{geomeq}) can
be written in coefficient forms. Introducing local coordinates $%
(x^{i},y^{\alpha })\in E,$ for Lie algebroid structure functions $(\rho
_{\alpha }^{i},C_{\alpha \beta }^{\gamma }),$ and choosing a basis $\{%
\mathcal{X}_{\alpha },\mathcal{V}_{\alpha }\}\in Sec(\mathcal{T}^{E}E),$ for all $\alpha ,$ we  have
\begin{eqnarray}
&&S\mathcal{X}_{\alpha }=\mathcal{V}_{\alpha },~S\mathcal{V}_{\alpha
}=0,~\bigtriangleup =y^{\alpha }\mathcal{V}_{\alpha },~E_{L}=y^{\alpha
}\partial L/\partial y^{\alpha }-L,  \label{form1} \\
&&\omega _{L}=\frac{\partial ^{2}L}{\partial y^{\alpha }\partial y^{\beta }}%
\mathcal{X}^{\alpha }\wedge \mathcal{V}^{\beta }+\frac{1}{2}(\rho _{\beta
}^{i}\frac{\partial ^{2}L}{\partial x^{i}\partial y^{\alpha }}-\rho _{\alpha
}^{i}\frac{\partial ^{2}L}{\partial x^{i}\partial y^{\beta }}+C_{\alpha
\beta }^{\gamma }\frac{\partial L}{\partial y^{\gamma }})\mathcal{X}^{\alpha
}\wedge \mathcal{X}^{\beta }.  \notag
\end{eqnarray}%
As a vertical endomorphism (equivalently, tangent structure) can be used the operator $S:=\mathcal{X}^{\alpha }\otimes \mathcal{V}_{\alpha }.$

The Euler--Lagrange section associated with $L$ is given by $\Gamma
_{L}=y^{\alpha }\mathcal{X}_{\alpha }+\varphi ^{\alpha }\mathcal{V}_{\alpha
},$ when functions $\varphi ^{\alpha }(x^{i},y^{\beta })$ solve this system
of linear equations%
\begin{equation*}
\varphi ^{\beta }\frac{\partial ^{2}L}{\partial y^{\beta }\partial y^{\alpha
}}+y^{\beta }\left( \rho _{\beta }^{i}\frac{\partial ^{2}L}{\partial
x^{i}\partial y^{\alpha }}+C_{\alpha \beta }^{\gamma }\frac{\partial L}{%
\partial y^{\gamma }}\right) -\rho _{\alpha }^{i}\frac{\partial L}{\partial
x^{i}}=0.
\end{equation*}%
The condition of regularity is equivalent to non--degeneration of the
Hessian
\begin{equation}
\varpi _{\alpha \beta }:=\frac{\partial ^{2}L}{\partial y^{\alpha }\partial
y^{\beta }},~|\varpi _{\alpha \beta }|=\det |\varpi _{\alpha \beta }|\neq 0.
\label{hessian}
\end{equation}%
For regular configurations, we can express the semi--spray vector as
\begin{equation}
\varphi ^{\varepsilon }=\varpi ^{\varepsilon \beta }(\rho _{\beta }^{i}\frac{%
\partial L}{\partial x^{i}}-\rho _{\alpha }^{i}\frac{\partial ^{2}L}{%
\partial x^{i}\partial y^{\beta }}y^{\alpha }-C_{\beta \alpha }^{\gamma }%
\frac{\partial L}{\partial y^{\gamma }}y^{\alpha }),  \label{semispray}
\end{equation}%
where $\varpi ^{\alpha \beta }$ is inverse to $\varpi _{\alpha \beta }.$ If
the condition $\left[ \bigtriangleup ,\Gamma _{L}\right] _{E}= \Gamma _{L}$
is satisfied, the section $\Gamma _{L}$ transforms into a spray which states
that the functions $\varphi ^{\beta }$ are homogenous of degree $2$ on $%
y^{\beta }.$ A curve $c(\tau )=(x^{i}(\tau ),y^{\alpha }(\tau ))\in E$ for a
real parameter $\tau $ defines a solution of the Euler--Lagrange equations
for $L$ if
\begin{equation}
\frac{dx^{i}}{d\tau }=\rho _{\alpha }^{i}y^{\alpha }\mbox{ \ and \ }\frac{d}{%
d\tau }\left( \frac{\partial L}{\partial y^{\alpha }}\right) +y^{\beta
}C_{\alpha \beta }^{\gamma }\frac{\partial L}{\partial y^{\gamma }}-\rho
_{\alpha }^{i}\frac{\partial L}{\partial x^{i}}=0.  \label{eleq}
\end{equation}

Similarly to the model of Lagrange mechanics on Lie algebroids defined by
equations (\ref{geomeq}) and (\ref{eleq}), it is possible to elaborate
Hamilton/symplectic geomerizations, see details in \cite{dleon1,cortes}.
However, in both cases, it is not clear how some versions of Perelman
functionals for geometric flows should be performed if we restrict our
constructions only to Cartan's symplectic forms and Lagrangian energy (\ref%
{cartvar}) and related equations (\ref{geomeq}).

\section{Lagrangians on Lie Algebroids \& N--Connections}

\label{s3} In order to elaborate Lagrange--Ricci evolution models on $TM$
and nonholonomic manifolds, we used \cite{vricci1,vricci2} a geometrization
of mechanics in terms of canonical nonlinear and linear connections defined
by a regular Lagrangian $L.$ This section is devoted to a brief introduction
into the geometry of nonlinear connections on Lie algebroids, see former
constructions \cite{kern,valg1,valg2}.

\subsection{N--connections and prolongations of Lie algebroids}

A nonlinear connection, N--connection, structure for a vector bundle $P$ %
\cite{ehresmann} can be defined as a Whitney sum $\mathbf{N}:TP=hTP\oplus
vTP $. A couple $\mathbf{P}:=(P,\mathbf{N})$ is called a \textit{%
nonholonomic vector bundle} (equivalently, vector N--bundle, with
conventional horizontal, h, and vertical, v, splitting/decomposition).%
\footnote{%
As some particular cases, we can take $P=TM$, for a tangent bundle, or to consider $P=V$, for as a (semi)\ Riemannian nonholonomic manifold with non--integrable $h$--$v$--splitting as we consider in our works on classical and quantum gravity \cite{vncricci,vrflg,vqgr2}), when ''boldface'' letters are used for spaces and/or geometric objects defined/adapted to spaces with N--connection structure.}

N--connections can be similarly introduced on prolongation Lie algebroids via a corresponding $h$--$v$--splitting,
\begin{equation}
\mathcal{N}:\mathcal{T}^{E}\mathbf{P}=h\mathcal{T}^{E}P\oplus v\mathcal{T}%
^{E}P.  \label{nonalg}
\end{equation}%
Such a bundle, and Lie algebroid, morphism $\mathcal{N}:\mathcal{T}^{E}%
\mathbf{P\rightarrow }\mathcal{T}^{E}\mathbf{P,}$ with $\mathcal{N}^{2}=id,$
defines an almost product structure on $\ ^{P}\pi:\ TP\rightarrow P$ for a
smooth map on $TP\backslash \{0\},$ were $\{0\}$ denotes the set of null
sections. A N--connection induces $h$- and $v$--projectors for every element
$\overline{z}=(p,b,v)\in $ $\mathcal{T}^{E}P,$ when $h(\overline{z})=~^{h}z$
and $v(z)=~^{v}z,$ for $h=\frac{1}{2}(id+\mathcal{N})$ and $v=\frac{1}{2}(id-%
\mathcal{N}).$ These operators define, respectively, the $h$- and $v$%
--subspaces, $h\mathcal{T}^{E}P=\ker (id-\mathcal{N})$ and $v\mathcal{T}%
^{E}P=\ker (id+\mathcal{N}).$

\begin{definition}
A Lie distinguished algebroid (d--algebroid) $\widetilde{\mathcal{E}}=(%
\mathbf{E},\left\lfloor \cdot,\cdot\right\rfloor,\rho),$ is defined for a
nonholonomic vector bundle $\mathbf{E}$ endowed with N--connection structure
$\mathbf{N}$. The prolongations of a Lie algebroid $\mathcal{E}$ over a
nonholonomic bundle $\mathbf{P}:=(P,\mathbf{N})$ is also a Lie d--algebroid.%
\footnote{%
In our former works, see \cite{valg1} and references therein, we wrote for a
Lie d--algebroid $~^{N}\mathcal{E}=(E,\mathbf{N},\left\lfloor \cdot ,\cdot
\right\rfloor ,\rho ),$ when left low/up indices are used as abstract labels
for some geometric objects and spaces.}
\end{definition}

Locally, N--connections are determined respectively by their coefficients $%
\mathbf{N=\{}N_{\alpha }^{A}\}$ and $\mathcal{N}\mathbf{=\{}\mathcal{N}%
_{\alpha }^{A}\},$ when
\begin{equation}
\mathbf{N=}N_{i}^{A}(x^{k},u^{B})dx^{i}\otimes \partial _{A}\mbox{ \ and \ }%
\mathcal{N}=\mathcal{N}_{\alpha }^{A}\mathcal{X}^{\alpha }\otimes \mathcal{V}%
_{A}.  \label{nonlalg}
\end{equation}%
Such structures on $TP$ and $\mathcal{T}^{E}\mathbf{P}$ are compatible if $%
\mathcal{N}_{\alpha }^{A}=N_{i}^{A}\rho _{\alpha }^{i}.$ Using $\mathcal{N}%
_{\alpha }^{A},$ we can generate sections $\delta _{\alpha }:=\mathcal{X}%
_{\alpha }-\mathcal{N}_{\alpha }^{A}\mathcal{V}_{A}$ as a local basis of $h%
\mathcal{T}^{E}P.$ In general, this allows us to define a N--adapted frame structure%
\begin{equation}
~\mathbf{e}_{\overline{\alpha }}:=\{\delta _{\alpha }=\mathcal{X}_{\alpha }-%
\mathcal{N}_{\alpha }^{C}\mathcal{V}_{C},\mathcal{V}_{A}\},\mbox{ \ on \ }%
\mathcal{T}^{E}\mathbf{P},  \label{dderalg}
\end{equation}%
and its dual
\begin{equation}
\mathbf{e}^{\overline{\beta }}:=\{\mathcal{X}^{\alpha },\delta ^{B}=\mathcal{%
V}^{B}+\mathcal{N}_{\gamma }^{B}\mathcal{V}^{\gamma }\},  \label{ddifalg}
\end{equation}%
where the ''overlined'' small Greek indices split in the form $\overline{%
\alpha }=(\alpha ,A)$ if an arbitrary vector bundles $\mathbf{P}$ is
considered, or $\overline{\alpha }=(\alpha ,\alpha )$ if $\mathbf{P=E.}$ The
N--adapted bases (\ref{dderalg}) satisfy certain nonholonomy relations,
\begin{equation*}
~\mathbf{e}_{\overline{\alpha }}~\mathbf{e}_{\overline{\beta }}-~\mathbf{e}_{%
\overline{\beta }}~\mathbf{e}_{\overline{\alpha }}=W_{\overline{\alpha }%
\overline{\beta }}^{\overline{\gamma }}\mathbf{e}_{\overline{\gamma }},
\end{equation*}%
with nontrivial anholonomy coefficients $W_{\overline{\alpha }\overline{%
\beta }}^{\overline{\gamma }}=\{C_{\alpha \beta }^{\gamma },\Omega _{\alpha
\beta }^{C},\partial _{B}\mathcal{N}_{\alpha }^{C}\}.$ Such values are
determined both by $N_{\alpha }^{A}$ and Lie algebroid structure constants $%
C_{\alpha \beta }^{\gamma }.$ The corresponding generalized Lie brackets are
defined by relations
\begin{equation*}
\left\lfloor \delta _{\alpha },\delta _{\beta }\right\rfloor ^{\pi
}=C_{\alpha \beta }^{\gamma }\delta _{\gamma }+\Omega _{\alpha \beta }^{C}%
\mathcal{V}_{C},~\left\lfloor \delta _{\alpha },\mathcal{V}_{B}\right\rfloor
^{\pi }=(\partial _{B}\mathcal{N}_{\alpha }^{C})\mathcal{V}%
_{C},~~\left\lfloor \mathcal{V}_{A},\mathcal{V}_{B}\right\rfloor ^{\pi }=0.
\end{equation*}%
In these formulas, $\Omega _{\alpha \beta }^{C}=\delta _{\beta }\mathcal{N}%
_{\alpha }^{C}-\delta _{\alpha }\mathcal{N}_{\beta }^{C}+C_{\alpha \beta }^{\gamma }\mathcal{N}_{\gamma }^{C}$ are the N--adapted coefficients of the Neigenhuis tensor $~^{h}N$ of the operator $h,$ {\small
\begin{equation*}
\ ^{h}N(\cdot ,\cdot) =\left\lfloor h\cdot ,h\cdot \right\rfloor ^{\pi
}-h\left\lfloor h\cdot ,\cdot \right\rfloor ^{\pi }-h\left\lfloor \cdot
,h\cdot \right\rfloor ^{\pi }+h^{2}\left\lfloor h\cdot ,h\cdot \right\rfloor
^{\pi } =-\frac{1}{2}\Omega _{\alpha \beta }^{C}\mathcal{X}^{\alpha}
\wedge \mathcal{X}^{\beta }\otimes \mathcal{V}_{C},
\end{equation*}%
}
which (by definition) is considered to be the curvature of N--connection $%
\mathcal{N}_{\alpha }^{A}.$

It should be noted that for $\mathbf{P=E,}$ the above formulas for Lie
d--algebroid $\mathcal{T}^{E}\mathbf{E}$ mimic on sections of $E$ the
geometry of tangent bundles and/or nonholonomic manifolds of even dimension,
endowed with N--connection structure (on applications in modern classical
and quantum gravity, with various modifications, and nonholonomic Ricci flow
theory, see Refs. \cite{vrflg,vricci1,vncricci}). If $\mathbf{P\neq E,}$ we
model nonholonomic vector bundle and generalized Riemann geometries on
sections of $\mathcal{T}^{E}\mathbf{P.}$

\subsection{Linear connections and metrics on $\mathcal{T}^{E}\mathbf{P}$}

The Levi--Civita connection $~^{\varpi }\nabla $ (\ref{lccon}) on $\mathcal{E%
}$ is not adapted to a N--connection structure on $\mathcal{T}^{E}\mathbf{P}$%
. We have to introduce into consideration another classes of linear
connections which would involve the $h$-$v$--splitting for $\mathcal{T}^{E}%
\mathbf{P}$.

\begin{definition}
A distinguished connection, d--connection, $\mathcal{D}$ on $\mathcal{T}^{E}%
\mathbf{P}$ is a linear connection preserving under parallelism the
N--connection (\ref{nonalg}).
\end{definition}

The N--adapted components $\mathbf{\Gamma }_{\ \overline{\beta }\overline{%
\gamma }}^{\overline{\alpha }}=\left( \mathbf{L}_{\beta \gamma }^{\alpha },%
\mathbf{L}_{B\gamma \;}^{A};\mathbf{B}_{\beta C}^{\alpha },\mathbf{B}%
_{BC}^{A}\right) $ of a covariant operator $\mathcal{D}_{\overline{\alpha }%
}=(~\mathbf{e}_{\overline{\alpha }}\rfloor \mathcal{D}),$ where $\rfloor $
is the interior product, are computed following equations $\mathcal{D}_{%
\overline{\alpha }}\mathbf{e}\overline{_{\beta }}=\mathbf{\Gamma }_{\
\overline{\alpha }\overline{\beta }}^{\overline{\gamma }}\mathbf{e}_{%
\overline{\gamma }},$ or $\mathbf{\Gamma }_{\ \alpha \beta }^{\gamma
}=\left( \mathcal{D}_{\overline{\alpha }}\mathbf{e}_{\overline{\beta }%
}\right) \rfloor \mathbf{e}^{\overline{\gamma }}.$ The h-- and v--covariant
derivatives are respectively $h\mathcal{D}=\{\mathcal{D}_{\gamma }=\left(
\mathbf{L}_{\beta \gamma }^{\alpha },\mathbf{L}_{B\gamma \;}^{A}\right) \}$
and $\ v\mathcal{D}=\{\mathcal{D}_{C}=\left( \mathbf{B}_{\beta C}^{\alpha },%
\mathbf{B}_{BC}^{A}\right) \}$, where $\mathbf{L}_{\beta \gamma }^{\alpha
}:=(\mathcal{D}_{\gamma }\delta _{\beta })\rfloor \mathcal{X}^{\alpha },%
\mathbf{L}_{B\gamma }^{A}:=(\mathcal{D}_{\gamma }\mathcal{V}_{B})\rfloor
\delta ^{A},\mathbf{B}_{\beta C}^{\alpha }:=(\mathcal{D}_{C}\delta _{\beta
})\rfloor \mathcal{X}^{\alpha },\mathbf{B}_{BC}^{A}:=(\mathcal{D}_{C}%
\mathcal{V}_{B})\rfloor \delta ^{A}$ are computed for N--adapted bases (\ref%
{dderalg}) and (\ref{ddifalg}).

Using rules of absolute differentiation (\ref{extercalc}) for N--adapted
bases $~\mathbf{e}_{\overline{\alpha }}:=\{\delta _{\alpha },\mathcal{V}%
_{A}\}$ and $\mathbf{e}^{\overline{\beta }}:=\{\mathcal{X}^{\alpha },
\delta^{B}\}$ and the d--connection 1--form $\mathbf{\Gamma }_{\ \overline{%
\alpha }}^{\overline{\gamma }}:=\mathbf{\Gamma }_{\ \overline{\alpha }%
\overline{\beta }}^{\overline{\gamma }}\mathbf{e}^{\overline{\beta }},$ we
can compute the torsion and curvature 2--forms on $\mathcal{T}^{E}\mathbf{P}%
: $

Let us consider sections $\overline{x},\overline{y},\overline{z}$ \ of $%
\mathcal{T}^{E}\mathbf{P,}$ were (for instance) $\overline{z}=z^{\overline{%
\alpha }}\mathbf{e}_{\overline{\alpha }}=z^{\alpha }\delta _{\alpha }+z^{A}%
\mathcal{V}_{A}.$ The torsion of d--connection $\mathcal{D}$, $\mathcal{T}(%
\overline{x},\overline{y})=\mathcal{D}_{\overline{x}}\overline{y}- \mathcal{D%
}_{\overline{y}}\overline{x}+\left\lfloor \overline{x},\overline{y}%
\right\rfloor ^{\pi }$ considered as a 2--form is defined as $\mathcal{T}^{%
\overline{\alpha }}:=\mathcal{D}\mathbf{e}^{\overline{\alpha }}=d\mathbf{e}^{%
\overline{\alpha }}+\mathbf{\Gamma }_{\ \overline{\beta }}^{\overline{\alpha
}}\wedge \mathbf{e}^{\overline{\beta }}$.  Following a straightforward
N--adapted differential form calculus, we prove

\begin{theorem}
The $h$--$v$--coefficients of torsion, \newline
$\mathcal{T}^{\overline{\alpha }}=\{\mathbf{T}_{\ \overline{\beta }\overline{%
\gamma }}^{\overline{\alpha }}\}=\{\mathbf{T}_{\ \beta \gamma }^{\alpha },%
\mathbf{T}_{\ \beta A}^{\alpha },\mathbf{T}_{\ \beta \gamma }^{A},\mathbf{T}%
_{\ B\alpha }^{A},\mathbf{T}_{\ BC}^{A}\},$ are
\begin{eqnarray}
\mathbf{T}_{\ \beta \gamma }^{\alpha } &=&\mathbf{L}_{\ \beta \gamma
}^{\alpha }-\mathbf{L}_{\ \gamma \beta }^{\alpha }+C_{\ \beta \gamma
}^{\alpha },\ \mathbf{T}_{\ \beta A}^{\alpha }=-\mathbf{T}_{\ A\beta
}^{\alpha }=\mathbf{B}_{\ \beta A}^{\alpha },\ \mathbf{T}_{\ \beta \alpha
}^{A}=\Omega _{\ \beta \alpha }^{A},\   \notag \\
\mathbf{T}_{\ B\alpha }^{A} &=&\frac{\partial \mathcal{N}_{\alpha }^{A}}{%
\partial u^{B}}-\mathbf{L}_{\ B\alpha }^{A},\ \mathbf{T}_{\ BC}^{A}=\mathbf{B%
}_{\ BC}^{A}-\mathbf{B}_{\ CB}^{A}.  \label{dtors}
\end{eqnarray}
\end{theorem}

The curvature of $\mathcal{D}$, $\mathcal{R}(\overline{x},\overline{y})%
\overline{z}:=\left( \mathcal{D}_{\overline{x}}~\mathcal{D}_{\overline{y}}-~%
\mathcal{D}_{\overline{y}}\mathcal{D}_{\overline{x}}-\mathcal{D}%
_{\left\lfloor \overline{x},\overline{y}\right\rfloor ^{\pi }}\right)
\overline{z}$, also can be considered/ computed as a 2--form,
\begin{eqnarray}
\mathcal{R}_{~\overline{\beta }}^{\overline{\alpha }}&:=&\mathcal{D}\mathbf{%
\Gamma }_{\ \overline{\beta }}^{\overline{\alpha }}=d\mathbf{\Gamma }_{\
\overline{\beta }}^{\overline{\alpha }}-\mathbf{\Gamma }_{\ \overline{\beta }%
}^{\overline{\gamma }}\wedge \mathbf{\Gamma }_{\ \overline{\gamma }}^{%
\overline{\alpha }}=\mathbf{R}_{\ \overline{\beta }\overline{\gamma }%
\overline{\delta }}^{\overline{\alpha }}\mathbf{e}^{\overline{\gamma }%
}\wedge \mathbf{e}^{\overline{\delta }},  \label{curv1a} \\
\mbox{ where } \mathbf{R}_{\ ~\overline{\beta }\overline{\gamma }\overline{%
\delta }}^{\overline{\alpha }}&=&\mathbf{e}_{\overline{\delta }}\mathbf{%
\Gamma }_{\ \overline{\beta }\overline{\gamma }}^{\overline{\alpha }}-%
\mathbf{e}_{\overline{\gamma }}\ \mathbf{\Gamma }_{\ \overline{\beta }%
\overline{\delta }}^{\overline{\alpha }}+\mathbf{\Gamma }_{\ \overline{\beta
}\overline{\gamma }}^{\overline{\varphi }}\ \mathbf{\Gamma }_{\ \overline{%
\varphi }\overline{\delta }}^{\overline{\alpha }}-\mathbf{\Gamma }_{\
\overline{\beta }\overline{\delta }}^{\overline{\varphi }}\ \mathbf{\Gamma }%
_{\ \overline{\varphi }\gamma }^{\overline{\alpha }}+\mathbf{\Gamma }_{\
\overline{\beta }\overline{\varphi }}^{\overline{\alpha }}W_{\overline{%
\gamma }\overline{\delta }}^{\overline{\varphi }}.  \notag
\end{eqnarray}
This results in a proof of
\begin{theorem}
\label{thasa1}The curvature of d--connection of $\mathcal{D}$, \newline
$\mathcal{R}_{~\overline{\beta }}^{\overline{\alpha }}=\{\mathbf{R}_{\
\overline{\beta }\overline{\gamma }\overline{\delta }}^{\overline{\alpha }%
}\}=\{\mathbf{R}_{\ \varepsilon \beta \gamma }^{\alpha }\mathbf{,R}_{\
B\beta \gamma }^{A}\mathbf{,R}_{\ \varepsilon \beta A}^{\alpha }\mathbf{,R}%
_{\ B\beta A}^{C}\mathbf{,R}_{\ \beta BA}^{\alpha },\mathbf{R}_{\ BEA}^{C}\}$
is characterized by N--adapted coefficients {\small
\begin{eqnarray}
\mathbf{R}_{\ \varepsilon \beta \gamma }^{\alpha } &=&\delta _{\gamma }%
\mathbf{L}_{\ \varepsilon \beta }^{\alpha }-\delta _{\beta }\mathbf{L}_{\
\varepsilon \gamma }^{\alpha }+\mathbf{L}_{\ \varepsilon \beta }^{\mu }%
\mathbf{L}_{\ \mu \gamma }^{\alpha }-\mathbf{L}_{\ \varepsilon \gamma }^{\mu
}\mathbf{L}_{\ \mu \beta }^{\alpha }+\mathbf{L}_{\ \varepsilon \varphi
}^{\alpha }C_{\beta \gamma }^{\varphi }-\mathbf{B}_{\ \varepsilon A}^{\alpha
}\Omega _{\ \gamma \beta }^{A},  \notag \\
\mathbf{R}_{\ B\beta \gamma }^{A} &=&\delta _{\gamma }\mathbf{L}_{\ B\beta
}^{A}-\delta _{\beta }\mathbf{L}^{A}_{B\gamma}+ \mathbf{L}^{C}_{B\beta}%
\mathbf{L}_{\ C\gamma }^{A}-\mathbf{L}^{C}_{B\gamma}\mathbf{L}_{\ C\beta
}^{A}+\mathbf{L}^{A}_{B\varphi }C_{\beta \gamma }^{\varphi }-\mathbf{B}_{\
BC}^{A}\Omega _{\gamma \beta }^{C},  \notag \\
\mathbf{R}_{\ \varepsilon \beta A}^{\alpha } &=&\mathcal{V}_{A}\mathbf{L}_{\
\varepsilon \beta }^{\alpha }-\mathcal{D}_{\beta }\mathbf{B}_{\varepsilon
A}^{\alpha }+\mathbf{B}_{\ \varepsilon B}^{\alpha }\mathbf{T}_{\ \beta
A}^{B},  \label{dcurv} \\
\mathbf{R}_{\ B\gamma A}^{C} &=&\mathcal{V}_{A}\mathbf{L}_{\ bk}^{c}-%
\mathcal{D}_{\gamma }\mathbf{B}_{\ BA}^{C}+\mathbf{B}_{\ BD}^{C}\mathbf{T}%
_{\ \gamma A}^{D},  \notag \\
\mathbf{R}_{\ \beta BA}^{\alpha } &=&\mathcal{V}_{A}\mathbf{B}_{\ \beta
B}^{\alpha }-\mathcal{V}_{B}\mathbf{B}_{\beta C}^{\alpha }+\mathbf{B}_{\
\beta B}^{\tau }\mathbf{B}_{\ \tau C}^{\alpha }-\mathbf{B}_{\ \beta C}^{\tau
}\mathbf{B}_{\ \tau B}^{\alpha },  \notag \\
\mathbf{R}_{\ ECB}^{A} &=&\mathcal{V}_{E}\mathbf{B}_{\ BC}^{A}-\mathcal{V}%
_{C}\mathbf{B}_{\ BE}^{A}+\mathbf{B}_{\ BC}^{F}\mathbf{B}_{\ FE}^{A}-\mathbf{%
B}_{\ BE}^{F}\mathbf{B}_{\ FC}^{A}.  \notag
\end{eqnarray}%
}
\end{theorem}
We note that in above first two formulas  the terms $\mathbf{L}_{\ \varepsilon \varphi }^{\alpha
}C_{\beta \gamma }^{\varphi }$ and $\mathbf{L}_{\ B\varphi }^{A}C_{\beta
\gamma }^{\varphi },$ respectively, transform in
zero for a trivial Lie algebroid commutator structure when $C_{\beta
\gamma}^{\varphi }=0.$ In such a case, the geometry of $\mathcal{T}^{E}%
\mathbf{P}$ endowed with N--connection structure $\mathcal{N}_{\gamma }^{C}$
mimics a similar one for the associated vector bundle $\mathbf{P}$ with a
nontrivial $N_{i}^{\alpha }.$ Using prolongations of Lie algebroids on
fibration maps, we model tangent bundle geometries but not in a complete
equivalent form because there are differences in chosen nonholnomic
structures and torsions and curvatures of d--connections.

\begin{corollary}
\label{corolricci} The Ricci tensor of $\mathcal{D}$, $\mathcal{R}ic=\{%
\mathbf{R}_{\overline{\alpha }\overline{\beta }}:=\mathbf{R}_{\ \overline{%
\alpha }\overline{\beta }\overline{\gamma }}^{\overline{\gamma }}\}$, is
characterized by N--adapted coefficients{\small
\begin{equation}
\mathbf{R}_{\overline{\alpha }\overline{\beta }}=\{\mathbf{R}_{\alpha \beta
}:=\mathbf{R}_{\ \alpha \beta \chi }^{\gamma },\ \mathbf{R}_{\alpha A}:=-%
\mathbf{R}_{\ ~\alpha \gamma A}^{\gamma },\ \mathbf{R}_{A\alpha }:=\mathbf{R}%
_{\ ~A\alpha B}^{B},\ \mathbf{R}_{AB}:=\mathbf{R}_{\ ~ABC}^{C}\}.
\label{driccialg}
\end{equation}%
}
\end{corollary}

\begin{proof}
The formulas for $h$--$v$--components (\ref{driccialg}) are respective
contractons of the coefficients (\ref{dcurv}). $\square $
\end{proof}

\vskip5pt

\begin{definition}
A metric structure on $\mathcal{T}^{E}\mathbf{P}$ is defined by a
nondegenerate symmetric second rank tensor $\overline{\mathbf{g}}=\{$ $%
\mathbf{g}_{\overline{\alpha }\overline{\beta }}\}.$ Such a tensor is called
a distinguished metric, i.e. a d--metric, if its coefficients are defined
with respect to tensor products of N--adapted frames (\ref{ddifalg}),%
\begin{equation}
\overline{\mathbf{g}}=\mathbf{g}_{\overline{\alpha }\overline{\beta }}%
\mathbf{e}^{\overline{\beta }}\otimes \mathbf{e}^{\overline{\beta }}=\
\mathbf{g}_{\alpha \beta }\ \mathcal{X}^{\alpha }\otimes \mathcal{X}^{\beta
}+\ \mathbf{g}_{AB}\ \delta ^{A}\otimes \delta ^{B}.  \label{dm}
\end{equation}
\end{definition}

We can define the inverse d--metric $\mathbf{g}^{\overline{\alpha }\overline{%
\beta }}$ and inverse N--adapted h--metric, $\ \mathbf{g}^{\alpha \beta },$
and v--metric, $\mathbf{g}^{AB},$ by inverting respectively the matrix $%
\mathbf{g}_{\overline{\alpha }\overline{\beta }}$ and its bloc components, $%
\ \mathbf{g}_{\alpha \beta }$ and $\ \mathbf{g}_{AB}.$

The scalar curvature $\ ^{s}\mathbf{R}$ of $\ \mathcal{D}$ is by definition
\begin{equation}
\ ^{s}\mathbf{R}:=\mathbf{g}^{\overline{\alpha }\overline{\beta }}\mathbf{R}%
_{\overline{\alpha }\overline{\beta }}=\mathbf{g}^{\alpha \beta }\mathbf{R}%
_{\alpha \beta }+\mathbf{g}^{AB}\mathbf{R}_{AB}.  \label{sdcurv}
\end{equation}%
Using (\ref{driccialg}) and (\ref{sdcurv}), we can compute the Einstein
tensor $\mathbf{E}_{\overline{\alpha }\overline{\beta }}$ of $\mathcal{D},$
\begin{equation}
\mathbf{E}_{\overline{\alpha }\overline{\beta }}\doteqdot \mathbf{R}_{%
\overline{\alpha }\overline{\beta }}-\frac{1}{2}\mathbf{g}_{\overline{\alpha
}\overline{\beta }}\ \ ^{s}\mathbf{R}.  \label{enstdt}
\end{equation}%
Such a tensor can be used for modeling effective gravity theories on
sections of $\mathcal{T}^{E}\mathbf{P}$ with nonholonomic frame structure %
\cite{vrflg,vqgr2,valg1,valg2}.

\subsection{Metric compatible geometries on Lie d--algebroids}

Additionally to torsion (\ref{dtors}) and curvature (\ref{dcurv}),
d--connections are characterized by nonmetricity field $\mathcal{Q}(%
\overline{y}):=\mathcal{D}_{\overline{y}}\overline{\mathbf{g}}$, when $%
\mathbf{Q}_{\ \overline{\alpha }\overline{\beta }}^{\overline{\gamma }}=%
\mathcal{D}^{\overline{\gamma }}\overline{\mathbf{g}}_{\overline{\alpha }%
\overline{\beta }}$.

\begin{proposition}
The condition of metric compatibility, $\mathcal{Q}=\mathcal{D}\overline{%
\mathbf{g}}=0$, splits into respective conditions for $h$-$v$--components, $%
\mathcal{D}_{\gamma }\mathbf{g}_{\alpha \beta }=0,\mathcal{D}_{A}\mathbf{g}%
_{\alpha \beta }=0,\mathcal{D}_{\gamma }\mathbf{g}_{AB}=0,\mathcal{D}_{C}%
\mathbf{g}_{AB}=0$.
\end{proposition}

\begin{proof}
It follows from a straightforward computation when the coefficients of
d--metric $\mathbf{g}_{\overline{\alpha }\overline{\beta }}$ (\ref{dm}) are
introduced into $\mathcal{D}_{\overline{y}}\overline{\mathbf{g}}=0$, for $%
\overline{y}=y^{\overline{\alpha }}\mathbf{e}_{\overline{\alpha }}=y^{\alpha
}\delta _{\alpha }+y^{A}\mathcal{V}_{A}.$ $\square $
\end{proof}

\vskip5pt

In this paper, we shall work with two ''preferred'' linear connections
completely defined by a d--metric structure $\overline{\mathbf{g}}$ on $%
\mathcal{T}^{E}\mathbf{P}$:

\begin{theorem}
There is a canonical d--connection $\widehat{\mathcal{D}}$ \ for which $%
\widehat{\mathcal{D}}\overline{\mathbf{g}}=0$ and $h$- and $v$-torsions (\ref%
{dtors}) are prescribed, respectively, to be with coefficients $\widehat{T}%
_{\ \beta \gamma }^{\alpha }=C_{\ \beta \gamma }^{\alpha }$ and $\widehat{T}%
_{\ BC}^{A}=0$ computed with respect to N--adapted frames.
\end{theorem}

\begin{proof}
We can check by straightforward computations that the conditions of this
theorem as satisfied if and only if $\widehat{\mathcal{D}}$ is taken with
N--adapted coefficients $\widehat{\mathbf{\Gamma }}_{\ \overline{\beta }%
\overline{\gamma }}^{\overline{\alpha }}=\left( \widehat{\mathbf{L}}_{\beta
\gamma }^{\alpha },\widehat{\mathbf{L}}_{B\gamma \;}^{A};\widehat{\mathbf{B}}%
_{\beta C}^{\alpha },\widehat{\mathbf{B}}_{BC}^{A}\right) $ for {\small
\begin{eqnarray}
\widehat{L}_{\beta \gamma }^{\alpha } &=&\frac{1}{2}\mathbf{g}^{\alpha \tau
}\left( \delta _{\gamma }\mathbf{g}_{\beta \tau }+\delta _{\beta }\mathbf{g}%
_{\gamma \tau }-\delta _{\tau }\mathbf{g}_{\beta \gamma }\right) + \frac{1}{2%
}\mathbf{g}^{\alpha \tau }\left( \mathbf{g}_{\beta \varepsilon }C_{\ \tau
\gamma }^{\varepsilon }+\mathbf{g}_{\gamma \varepsilon }C_{\ \tau \beta
}^{\varepsilon }-\mathbf{g}_{\tau \varepsilon }C_{\ \beta \gamma
}^{\varepsilon }\right),  \notag \\
\widehat{L}_{B\gamma }^{A} &=&\mathcal{V}_{B}(\mathcal{N}_{\gamma }^{A})+%
\frac{1}{2}\mathbf{g}^{AC}\left( \delta _{\gamma }\mathbf{g}_{BC}-\mathbf{g}%
_{DC}\mathcal{V}_{B}(\mathcal{N}_{\gamma }^{D})\ -\mathbf{g}_{DB}\mathcal{V}%
_{C}(\mathcal{N}_{\gamma }^{D})\right),  \label{candcon} \\
\widehat{B}_{\beta C}^{\alpha } &=&\frac{1}{2}\mathbf{g}^{\alpha \tau }%
\mathcal{V}_{C}\mathbf{g}_{\beta \tau },\ \widehat{B}_{BC}^{A}=\frac{1}{2}%
\mathbf{g}^{AD}\left( \mathcal{V}_{C}\mathbf{g}_{BD}+\mathcal{V}_{B}\mathbf{g%
}_{CD}-\mathcal{V}_{D}\mathbf{g}_{BC}\right) . \  \square   \notag
\end{eqnarray}%
}
\end{proof}
\vskip5pt

The nontrivial values of torsion of $\widehat{\mathcal{D}},$ i.e. N--adapted
coefficients $\widehat{T}_{\ \beta \gamma }^{\alpha },$ $\widehat{T}_{\
\beta A}^{\alpha },\widehat{T}_{\ \beta \alpha }^{A}$ and $\widehat{T}_{\
B\alpha }^{A},$ are computed by introducing the canonical d--connection
coefficients (\ref{candcon}) into formulas (\ref{dtors}).

\begin{theorem}[--Definition]
There is a metric compatible Levi--Civita connection $\overline{\nabla }$
which is completely defined by a d--metric structure $\overline{\mathbf{g}}$
on $\mathcal{T}^{E}\mathbf{P}$ following the condition of zero torsion, $%
~^{\nabla }\mathcal{T}^{\overline{\alpha }}=\{K_{\ \overline{\beta }%
\overline{\gamma }}^{\overline{\alpha }}\}=0.$
\end{theorem}

\begin{proof}
Such a connection $\overline{\nabla }=K_{\ \overline{\beta }\overline{\gamma
}}^{\overline{\alpha }}=\left( \overline{L}_{\beta \gamma }^{\alpha },%
\overline{L}_{B\gamma \;}^{A};\overline{B}_{\beta C}^{\alpha },\overline{B}%
_{BC}^{A}\right) $ can be defined with respect to N--adapted frames for the
same d--metric structure $\overline{\mathbf{g}}$ which is used for
constructing $\widehat{\mathcal{D}}$ (\ref{candcon}), but with additional
constraints that all torsion coefficients (\ref{dtors}) are zero. We can
verify via straightforward computations with respect to (\ref{dderalg}) and (%
\ref{ddifalg}) that the condition of theorem is satisfied by a distortion
relation
\begin{equation}
K_{\ \overline{\beta }\overline{\gamma }}^{\overline{\alpha }}=\widehat{%
\mathbf{\Gamma }}_{\ \overline{\beta }\overline{\gamma }}^{\overline{\alpha }%
}+\widehat{\mathbf{Z}}_{\ \overline{\beta }\overline{\gamma }}^{\overline{%
\alpha }},  \label{distrel1}
\end{equation}%
where the distortion tensor $\widehat{\mathcal{Z}}=\{\widehat{\mathbf{Z}}_{\
\alpha \beta }^{\gamma }\}$ is given by N--adapted coefficients {\small
\begin{eqnarray}
\ \widehat{\mathbf{Z}}_{\beta \gamma }^{A} &=&-\widehat{\mathbf{B}}_{\beta
B}^{\alpha }\mathbf{g}_{\alpha \gamma }\mathbf{g}^{AB}-\frac{1}{2}\Omega
_{\beta \gamma }^{A},~\widehat{\mathbf{Z}}_{B\gamma }^{\alpha }=\frac{1}{2}%
\Omega _{\alpha \gamma }^{C}\mathbf{g}_{CB}\mathbf{g}^{\beta \alpha }-\Xi
_{\beta \gamma }^{\alpha \tau }~\widehat{\mathbf{B}}_{\tau B}^{\beta },
\notag \\
\widehat{\mathbf{Z}}_{B\gamma }^{A} &=&~^{+}\Xi _{CD}^{AB}~\widehat{\mathbf{T%
}}_{\gamma B}^{C},\ \widehat{\mathbf{Z}}_{\gamma B}^{i}=\frac{1}{2}\Omega
_{\beta \gamma }^{A}\mathbf{g}_{CB}\mathbf{g}^{\beta \alpha }+\Xi _{\beta
\gamma }^{\alpha \tau }~\widehat{\mathbf{B}}_{\tau B}^{\beta },\ \widehat{%
\mathbf{Z}}_{\beta \gamma }^{\alpha }=0,  \label{deft} \\
\ \widehat{\mathbf{Z}}_{\beta B}^{A} &=&-~^{-}\Xi _{CB}^{AD}~\widehat{%
\mathbf{T}}_{\beta D}^{C},\ \widehat{\mathbf{Z}}_{BC}^{A}=0,\ \widehat{%
\mathbf{Z}}_{AB}^{\alpha }=-\frac{\mathbf{g}^{\alpha \beta }}{2}\left[
\widehat{\mathbf{T}}_{\beta A}^{C}\mathbf{g}_{CB}+\widehat{\mathbf{T}}%
_{\beta B}^{C}\mathbf{g}_{CA}\right] ,  \notag
\end{eqnarray}%
for $\ \Xi _{\beta \gamma }^{\alpha \tau }~=\frac{1}{2}(\delta _{\beta
}^{\alpha }\delta _{\gamma }^{\tau }-\mathbf{g}_{\beta \gamma }\mathbf{g}%
^{\alpha \tau })$ and $~^{\pm }\Xi _{CD}^{AB}=\frac{1}{2}(\delta
_{C}^{A}\delta _{D}^{B}\pm \mathbf{g}_{CD}\mathbf{g}^{AB}).$ }

The distortion coefficients (\ref{deft}) are such linear algebraic
combinations of coefficients of torsion of $\widehat{\mathcal{D}}$ that the
condition $\widehat{\mathbf{T}}_{\ \overline{\beta }\overline{\gamma }}^{%
\overline{\alpha }}=0$ is equivalent to $\widehat{\mathbf{Z}}_{\ \overline{%
\beta }\overline{\gamma }}^{\overline{\alpha }}=0,$ and inversely. So, we
can find a $h$--$v$--decomposition when $\overline{\Gamma }_{\ \overline{%
\beta }\overline{\gamma }}^{\overline{\alpha }}=\widehat{\mathbf{\Gamma }}%
_{\ \overline{\beta }\overline{\gamma }}^{\overline{\alpha }}$ even, in
general, $\overline{\nabla }\neq \widehat{\mathcal{D}}$;  such
connections are subjected to different rules of frame/coordinate transforms
on $\mathcal{T}^{E}\mathbf{P.}$ $\square $
\end{proof}

\vskip5pt

We emphasize that $\overline{\nabla }$ is not a d--connection and does not
preserve under parallelism the N--connection structure. Nevertheless, all
geometric data for $\left( \overline{\mathbf{g}},\overline{\nabla }\right) $
can be transformed equivalently into similar ones for $\left( \overline{%
\mathbf{g}},\widehat{\mathcal{D}},\mathcal{N}\right) $ when $\overline{%
\mathbf{g}}$ and $\mathcal{N}$ define a unique N--adapted splitting $%
\overline{\nabla }=\widehat{\mathcal{D}}+\widehat{\mathcal{Z}}.$\footnote{%
By geometric data, we consider any set of geometric tensors, forms,
connections etc and relevant field/evolution/ constraint equations which can
be used in a model of geometry and/or physical theory.}

\begin{corollary}
\label{coroffd}Any metric $\overline{\mathbf{g}}$ on $\mathcal{T}^{E}\mathbf{%
P}$ can be represented equivalently as a d--metric $\mathbf{g}_{\overline{%
\alpha }\overline{\beta }}$ (\ref{dm}) or, with respect to a local dual base
$dz^{\overline{\beta }}:=\{\mathcal{X}^{\alpha },\mathcal{V}^{B}\},$ in
generic off--diagonal form, $\mathbf{g}=g_{\overline{\alpha }\overline{\beta
}}dz^{\overline{\alpha }}\otimes dz^{\overline{\beta }},$ with
''non--boldface'' coefficients
\begin{equation}
\ g_{\overline{\alpha }\overline{\beta }}=\left[
\begin{array}{cc}
\ \mathbf{g}_{\alpha \beta }+\mathcal{N}_{\alpha }^{A}~\mathcal{N}_{\beta
}^{B}\ \ \mathbf{g}_{AB} & ~\mathcal{N}_{\beta }^{A}\ \ \mathbf{g}_{AC} \\
~\mathcal{N}_{\alpha }^{E}\ \mathbf{g}_{ED} & \mathbf{g}_{DC}%
\end{array}%
\right] .  \label{offd}
\end{equation}
\end{corollary}

\begin{proof}
A frame transform $\mathbf{e}^{\overline{\beta }}\rightarrow dz^{\overline{%
\beta }^{\prime }}$ is dual to matrix transform $\ \partial _{\overline{%
\alpha }^{\prime }}\rightarrow \mathbf{e}_{\overline{\alpha }}^{\ \overline{%
\alpha }^{\prime }}\partial _{\overline{\alpha }^{\prime }},$ with
\begin{equation}
\ \mathbf{e}_{\overline{\alpha }}^{\ \overline{\alpha }^{\prime }}=\left[
\begin{array}{cc}
\mathbf{\ e}_{\alpha }^{\ \alpha ^{\prime }} & \mathcal{N}_{\alpha }^{B}\
\mathbf{e}_{B}^{\ A^{\prime }} \\
0 & \mathbf{e}_{A}^{\ A^{\prime }}%
\end{array}%
\right] .  \label{ft}
\end{equation}%
Additionally, one should be considered some quadratic relations between
coefficients $\ g_{\overline{\alpha }\overline{\beta }}=e_{\overline{\alpha }%
}^{\ \overline{\alpha }^{\prime }}\ _{\overline{\beta }}^{\ \overline{\beta }%
^{\prime }}\eta _{\overline{\alpha }^{\prime }\overline{\beta }^{\prime }},$
for $\eta _{\overline{\alpha }^{\prime }\overline{\beta }^{\prime
}}=diag[\pm 1,...\pm 1]$ fixing a local signature for metric on $\mathcal{T}%
^{E}\mathbf{P}.$ A metric (\ref{offd}) is called generic off--diagonal
because it can not be diagonalized by coordinate transforms. $\square $
\end{proof}

\begin{remark}
Introducing $K_{\ \overline{\beta }\overline{\gamma }}^{\overline{\alpha }}=%
\widehat{\mathbf{\Gamma }}_{\ \overline{\beta }\overline{\gamma }}^{%
\overline{\alpha }}$ (\ref{candcon}) into formulas (\ref{dcurv}), (\ref%
{driccialg}) and (\ref{sdcurv}), we compute respectively the coefficients of
curvature, $\widehat{\mathbf{R}}_{\ \overline{\beta }\overline{\gamma }%
\overline{\delta }}^{\overline{\alpha }}$,\ Ricci tensor, $\widehat{\mathbf{R%
}}_{\overline{\alpha }\overline{\beta }}$, and scalar curvature, $~^{s}%
\widehat{\mathbf{R}}.$ The distortions $K=\widehat{\mathbf{\Gamma }}+\widehat{%
\mathbf{Z}}$ (\ref{distrel1}) allows us to compute the distorting tensors ($%
\widehat{\mathbf{Z}}_{\ \overline{\beta }\overline{\gamma }\overline{\delta }%
}^{\overline{\alpha }},\widehat{\mathbf{Z}}_{\overline{\alpha }\overline{%
\beta }}$ and $\ ^{s}\widehat{\mathbf{Z}})$ resulting in similar values for the
(pseudo) Riemannian geometry on $\mathcal{T}^{E}\mathbf{P}$ determined by $%
\left( \overline{\mathbf{g}},K\right) ,$ i.e. to define $R_{\ \overline{\beta }%
\overline{\gamma }\overline{\delta }}^{\overline{\alpha }},R_{\ \overline{%
\beta }\overline{\gamma }}$ and $~^{s}R.$
\end{remark}

We do not present all technical details and component formulas for
geometrical objects outlined in above Remark. As an example, we provide the
distortion relations for the Ricci tensor,
\begin{equation}
R_{\overline{\alpha }\overline{\beta }}=\widehat{\mathbf{R}}_{\overline{%
\alpha }\overline{\beta }}+\widehat{\mathbf{Z}}_{\overline{\alpha }\overline{%
\beta }},  \label{driccidist}
\end{equation}%
\begin{eqnarray*}
R_{\overline{\beta }\overline{\gamma }} &=&R_{\ ~\overline{\beta }\overline{%
\gamma }\overline{\alpha }}^{\overline{\alpha }}=\mathbf{e}_{\overline{%
\delta }}K_{\ \overline{\beta }\overline{\gamma }}^{\overline{\alpha }}-%
\mathbf{e}_{\overline{\gamma }}K_{\ \overline{\beta }\overline{\delta }}^{%
\overline{\alpha }}+K_{\ \overline{\beta }\overline{\gamma }}^{\overline{%
\varphi }}K_{\ \overline{\varphi }\overline{\delta }}^{\overline{\alpha }%
}-K_{\ \overline{\beta }\overline{\delta }}^{\overline{\varphi }}K_{\
\overline{\varphi }\gamma }^{\overline{\alpha }}+K_{\ \overline{\beta }%
\overline{\varphi }}^{\overline{\alpha }}W_{\overline{\gamma }\overline{%
\delta }}^{\overline{\varphi }}, \\
\widehat{\mathbf{Z}}_{\ \overline{\beta }\overline{\gamma }} &=&\widehat{%
\mathbf{Z}}_{\ \overline{\beta }\overline{\gamma }\overline{\alpha }}^{%
\overline{\alpha }}=\mathbf{e}_{\overline{\alpha }}\widehat{\mathbf{Z}}_{\
\overline{\beta }\overline{\gamma }}^{\overline{\alpha }}-\mathbf{e}_{%
\overline{\gamma }}\ \widehat{\mathbf{Z}}_{\ \overline{\beta }\overline{%
\alpha }}^{\overline{\alpha }}+\widehat{\mathbf{Z}}_{\ \overline{\beta }%
\overline{\gamma }}^{\overline{\varphi }}\ \widehat{\mathbf{Z}}_{\ \overline{%
\varphi }\overline{\alpha }}^{\overline{\alpha }}-\widehat{\mathbf{Z}}_{\
\overline{\beta }\overline{\alpha }}^{\overline{\varphi }}\ \widehat{\mathbf{%
Z}}_{\ \overline{\varphi }\overline{\gamma }}^{\overline{\alpha }}+ \\
&&\widehat{\mathbf{\Gamma }}_{\ \overline{\beta }\overline{\gamma }}^{%
\overline{\varphi }}\ \widehat{\mathbf{Z}}_{\ \overline{\varphi }\overline{%
\alpha }}^{\overline{\alpha }}-\widehat{\mathbf{\Gamma }}_{\ \overline{\beta
}\overline{\alpha }}^{\overline{\varphi }}\ \widehat{\mathbf{Z}}_{\
\overline{\varphi }\overline{\gamma }}^{\overline{\alpha }}+\widehat{\mathbf{%
Z}}_{\ \overline{\beta }\overline{\gamma }}^{\overline{\varphi }}\widehat{%
\mathbf{\Gamma }}_{\ \overline{\varphi }\overline{\alpha }}^{\overline{%
\alpha }}-\widehat{\mathbf{Z}}_{\ \overline{\beta }\overline{\alpha }}^{%
\overline{\varphi }}\ \widehat{\mathbf{\Gamma }}_{\ \overline{\varphi }%
\overline{\gamma }}^{\overline{\alpha }}+\widehat{\mathbf{Z}}_{\ \overline{%
\beta }\overline{\varphi }}^{\overline{\alpha }}W_{\overline{\gamma }%
\overline{\alpha }}^{\overline{\varphi }}.
\end{eqnarray*}%
Such values are defined with respect to N--adapted bases (\ref{dderalg}) and
(\ref{ddifalg}). Using fame transforms (\ref{ft}) and their dual ones
computed as inverse matrices $(\mathbf{e}_{\overline{\alpha }}^{\ \overline{%
\alpha }^{\prime }})^{-1}$, we can re--define the coefficients with respect
to coordinate bases. Coordinate formulas are important in the theory of
Ricci flows (allowing simplified proofs of a number of important results on
geometric evolution) and for constructing, in explicit form, exact solutions
in geometric mechanics and analogous gravity.

Finally, we note that all values on prolongation Lie algebroids are uniquely
determined by a d--metric $\mathbf{g}_{\overline{\alpha }\overline{\beta }}$
(\ref{dm}) (equivalently, by a generic off--diagonal $g_{\overline{\alpha }%
\overline{\beta }}$ (\ref{offd})) for a prescribed $\mathcal{N}_{\alpha
}^{B} $ (\ref{nonalg}). Elaborating a physical dynamical/ evolution model
for $\left( \overline{\mathbf{g}},\overline{\nabla }\right),$ the same
theory can be described in terms of data $\left( \overline{\mathbf{g}},%
\widehat{\mathcal{D}},\mathcal{N}\right).$ This property allows us to
simplify, for instance, the proofs of main results for Ricci flows on Lie
algebroids using similar results for (pseudo) Riemannian metrics and then
nonholonomically transforming the constructions into evolution of N--adapted
values.

\subsection{An extension of Kern--Matsumoto   approach   for  algebroid mechanics \& gravity}

Let us consider an alternative (second) approach to geometrization of
regular Lagrange mechanics \cite{kern,matsumoto} on Lie algebroids. All
constructions are described in terms of generalized metrics, adapted frames
and N- and d--connections. This is different from Cartan variables (\ref%
{cartvar}) and equations (\ref{geomeq1}) considered in section \ref{s2}.%
\footnote{%
To develop such ideas for mathematical relativity and geometric mechanics
and noncommutative and Lie algebroid modifications/ generalizations was
proposed in some our proposals for Projects and Marie Curies Fellowships in
2003--2004, see also papers \cite{valg1,valg2} and references therein.}

The goal of this section is to show that it is possible a setting when
canonical N-- and d--connections and d--metric on $\mathcal{T}^{E}\mathbf{E,}
$ for $\mathbf{P=E},$ considered in previous section, are derived from a
regular Lagrangian as a solution of Euler--Lagrange equations (\ref{eleq}).

\begin{lemma}
\label{lemsemispr}There is a N--connection $~^{q}\mathcal{N}:=-\mathcal{L}%
_{q}S$ defined by a semi--spray $q=y^{\alpha }\mathcal{X}_{\alpha
}+q^{\alpha }\mathcal{V}_{\alpha }$ and Lie derivative $\mathcal{L}_{q}$
acting on any $X\in Sec(TE)$ following formula $~^{q}\mathcal{N}%
(X)=-\left\lfloor q,SX\right\rfloor ^{\pi }+S\left\lfloor q,X\right\rfloor
^{\pi }.$
\end{lemma}

\begin{proof}
Consider the operators $S$ and $\bigtriangleup $ from (\ref{form1}) defining
the semi--spray via formula \ $Sq=\bigtriangleup .$ For local coordinates $%
(x^{i},u^{A})\rightarrow (x^{i},y^{\alpha })$ and $X=\mathcal{X}_{\alpha }$,
$\ $we compute $^{q}\mathcal{N}(\mathcal{X}_{\alpha })=-\left\lfloor q,S(%
\mathcal{X}_{\alpha })\right\rfloor ^{\pi }+S\left\lfloor q,\mathcal{X}%
_{\alpha }\right\rfloor ^{\pi }=\mathcal{X}_{\alpha }+(\partial _{\alpha
}q^{\beta }+y^{\gamma }C_{\gamma \alpha }^{\beta })\mathcal{V}_{\beta }$.
Using $~^{q}\mathcal{N}(\mathcal{V}_{\alpha })=-\mathcal{V}_{\alpha }$ and $%
~^{q}\mathcal{N}(\mathcal{X}_{\alpha })=\mathcal{X}_{\alpha }-2~^{q}\mathcal{%
N}_{\alpha }^{\gamma }(x,y)\mathcal{V}_{\gamma },$ we define the
N--connection coefficients $\mathcal{N}_{\alpha }^{\gamma }=-\frac{1}{2}%
\left( \partial _{\alpha }q^{\gamma }+y^{\beta }C_{\beta \alpha }^{\gamma }\right) $, see formulas (\ref{nonlalg}). $\square $
\end{proof}
\vskip5pt

To generate N--connections, we can use sections $\Gamma _{L}=y^{\alpha }%
\mathcal{X}_{\alpha }+\varphi ^{\alpha }\mathcal{V}_{\alpha },$ with $%
q^{\varepsilon }=\varphi ^{\varepsilon }(x^{i},y^{\beta })$ (\ref{semispray}).
\begin{theorem}
Any regular Lagrangian $L\in C^{\infty }(E)$ defines a canonical
N--connection on prolongation Lie algebroid $\mathcal{T}^{E}\mathbf{E}$,
\begin{equation}
\widetilde{\mathcal{N}}=~^{\varphi }\mathcal{N}=\{\widetilde{\mathcal{N}}%
_{\alpha }^{\gamma }=-\frac{1}{2}\left( \partial _{\alpha }\varphi ^{\gamma
}+y^{\beta }C_{\beta \alpha }^{\gamma }\right) \},  \label{canonnc}
\end{equation}
determined by semi--spray configurations encoding the solutions of
Euler--Lagrange equations (\ref{eleq}).
\end{theorem}

\begin{proof}
It is a straightforward consequence of above Lemma and (\ref{semispray}).$%
\square $
\end{proof}

\vskip5pt

The geometric data and dynamics of symplectic equations (\ref{geomeq1}) for
Cartan variables (\ref{cartvar}) can be encoded equivalently into a metric
compatible geometry on prolongation of Lie algebroid.

\begin{corollary}
A model of Lie algebroid geometry $(L:\widetilde{\mathcal{N}},\widetilde{%
\mathbf{g}},\widehat{\mathcal{D}})$ on $\mathcal{T}^{E}\mathbf{E,}$ for $%
\pi: E\rightarrow M,$ with prescribed algebroid structure functions $\rho
_{\alpha }^{i}(x^{k})$ and $C_{\beta \alpha }^{\gamma }(x^{k})$, is
canonically determined by a regular Lagrangian $L\in C^{\infty }(E)$.
\end{corollary}

\begin{proof}
It follows from such key steps in definition of fundamental geometric
objects. Using $L(x,y),$ we construct the canonical N--connection $%
\widetilde{\mathcal{N}}=\{\widetilde{\mathcal{N}}_{\alpha }^{\gamma }\}$ (%
\ref{canonnc}) and induced N--adapted frames (\ref{dderalg}) and (\ref%
{ddifalg}), respectively,
\begin{equation}
~\widetilde{\mathbf{e}}_{\overline{\alpha }}:=\{\delta _{\alpha }=\mathcal{X}%
_{\alpha }-\widetilde{\mathcal{N}}_{\alpha }^{\gamma }\mathcal{V}_{\gamma },%
\mathcal{V}_{\beta }\}\mbox{ \ and \ }\widetilde{\mathbf{e}}^{\overline{%
\beta }}:=\{\mathcal{X}^{\alpha },\delta ^{\beta }=\mathcal{V}^{\beta }+%
\widetilde{\mathcal{N}}_{\gamma }^{\beta }\mathcal{V}^{\gamma }\}.
\label{ddcanad}
\end{equation}%
At the next step, we construct a total metric of type (\ref{dm}), $\overline{%
\mathbf{g}}\rightarrow \widetilde{\mathbf{g}},$ as a Sasaki lift of Hessian $%
\varpi _{\alpha \beta }$ (\ref{hessian}), where
\begin{equation}
\widetilde{\mathbf{g}}:=\widetilde{\mathbf{g}}_{\overline{\alpha }\overline{%
\beta }}\mathbf{e}^{\overline{\beta }}\otimes \mathbf{e}^{\overline{\beta }%
}=\varpi _{\alpha \beta }\ \mathcal{X}^{\alpha }\otimes \mathcal{X}^{\beta
}+\ \varpi _{\alpha \beta }\ \delta ^{\alpha }\otimes \delta ^{\beta }.
\label{ldm}
\end{equation}%
Introducing the coefficients of d--metric (\ref{ldm}) into formulas (\ref%
{candcon}), we compute the coefficients of canonical $\widehat{\mathcal{D}}$%
, induced by $L.$ $\square $
\end{proof}

\vskip5pt

In general, we can use arbitrary frames of reference on $\mathcal{T}^{E}%
\mathbf{E,}$ when $\mathbf{e}_{\overline{\gamma }^{\prime }}=e_{\ \overline{%
\gamma }^{\prime }}^{\overline{\gamma }}\mathbf{\tilde{e}}_{\overline{\gamma
}}$ for any $\mathbf{\tilde{e}}_{\overline{\gamma }}$ (\ref{ddcanad}). Any
N--connection and/or metric structure $(\mathcal{N}\mathbf{,}\overline{%
\mathbf{g}})$ can be related to some canonical data $(\widetilde{\mathcal{N}}%
\mathbf{,}\widetilde{\mathbf{g}})$ determined by a regular Lagrangian $L\in
C^{\infty }(E),$ when $\overline{\mathbf{g}}_{\overline{\alpha }^{\prime }%
\overline{\beta }^{\prime }}=e_{\ \overline{\alpha }^{\prime }}^{\overline{%
\alpha }}e_{\ \overline{\beta }^{\prime }}^{\overline{\beta }}\mathbf{\tilde{%
g}}_{\overline{\alpha }\overline{\beta }}.$ It is necessary to solve an
algebraic quadratic system of equations in order to define $e_{\ \overline{%
\alpha }^{\prime }}^{\overline{\alpha }}$ from some prescribed data $%
\overline{\mathbf{g}}_{\overline{\alpha }^{\prime }\overline{\beta }^{\prime
}}$ and $\mathbf{\tilde{g}}_{\overline{\alpha }\overline{\beta }}.$

\begin{conclusion}
Via corresponding frame transforms and re--adapting nonholonomic
distributions on $\mathcal{T}^{E}\mathbf{E}$, we can model equivalently:

\begin{itemize}
\item any Lagrange mechanics determined by regular $L\in C^{\infty }(E)$ as
a Kern--Matsumoto model $(L:\widetilde{\mathcal{N}},\widetilde{\mathbf{g}},%
\widehat{\mathcal{D}})$ and corresponding Ricci tensor, $\widehat{\mathcal{R}%
}ic=\{\widehat{\mathbf{R}}_{\overline{\alpha }\overline{\beta }}\}$ (\ref%
{driccialg}), and scalar curvature $\ ^{s}\widehat{\mathbf{R}}$ (\ref{sdcurv}%
);

\item inversely, any off--diagonal metric $g_{\overline{\alpha }\overline{%
\beta }}$ (\ref{offd}) can be transformed via N--adapted frame transforms (%
\ref{ft}) into a d--metric $\mathbf{g}_{\overline{\alpha }\overline{\beta }}$
(\ref{dm}) (for a prescribed $L$, parametrized in a form $\widetilde{\mathbf{%
g}}$ (\ref{ldm})); we can model analogous gravity theories on algebroids as
effective Lagrange models.
\end{itemize}
\end{conclusion}

Following different approaches, algebroid models for analogous gravity and
matter field interactions are studied in Refs. \cite%
{vrflg,vqgr2,valg1,valg2,strobl1}. One of the most important problems for such theories is to provide a physical motivation for the type of linear
connection which should be chosen for constructing analogs of Einstein
equations of Lie algebroids and how solutions of gravitational filed
equations are related to the Euler--Lagrange equations in an effective
geometric mechanics. For  purposes of this paper, it is important to
consider the case of Einstein N--anholonomic spaces on Lie algebroid
prolongations defined by solutions of equations
\begin{equation}
\widehat{\mathbf{R}}_{\overline{\alpha }\overline{\beta }}=\lambda \mathbf{g}%
_{\overline{\alpha }\overline{\beta }}  \label{einstmeq}
\end{equation}%
where $\lambda =const.$ Such configurations are stationary ones in the
theory of Ricci flows and define the so--called Ricci solitons. In the case
of Lie algebroids determined by $\pi : E \rightarrow M,$ with local
coordinates $(x^{i},y^{\alpha }),$ the solutions for effective metrics
induced on sections are with Killing symmetries on $\partial /\partial
y^{\alpha },$ when Lie derivatives of $\mathbf{g}_{\overline{\alpha }%
\overline{\beta }}$ on such $y$--directions are zero, because all algebroid
geometric structures are defined as sections over $M$ (all structure
functions and coefficients of fundamental geometric objects depending on
local coordinates $x^{i}$).

\begin{claim}
Stationary (with respect to Ricci flow evolution of geometric structure on
Lie algebroids, see next section) effective Lagrange and/or analogous
gravitational models on a prolongation Lie algebroid $\mathcal{T}^{E}\mathbf{%
E}$ \ for $\pi :E\rightarrow M,$ $\dim E=n+m$ and $\dim M=n\geq 2,$ are
defined as Ricci soliton configurations (\ref{einstmeq}) with $m$ Killing
symmetries.
\end{claim}

For certain classes of smooth functions, such a Claim can be proven using
theorems on decoupling and integration of the Einstein--Yang--Mills--Higgs
equations, \cite{vrflg,veymheq}. Nevertheless, this Claim can not be proven
for all possible types of Lie algebroid configurations. The gravitational
and matter field equations on different curved spaces, including
constructions with Lie algebroids, are very sophisticate nonlinear systems
of partial differential equations. In general, such systems may have various
stochastic, fractional, chaos etc properties. This gives us a reason to
argue that following our experience a chosen class of Cauchy type and/or
stochastic etc flows can be modelled by a corresponding effective Lagrange
dynamics/ evolution of Lie algebroid configurations. We can not prove that
all physically important cases can be described via such models and it is
not possible to state some uniqueness criteria, completeness of solutions
etc.

\section{Lagrange--Ricci Evolution and Lie Algebroids}

\label{s4} Following Kern--Matsumoto geometrization of regular Lagrange
mechanics and analogous gravity models on Lie algebroids, we can consider
the problem of geometric flow evolution of such system as an explicit
example of a theory of Ricci flows on nonholonomic manifolds as we stated in
Refs. \cite{vricci1,vricci2,vncricci}. The goal of this section is to prove
that Lagrange--Ricci flows on $\mathcal{T}^{E}\mathbf{E}$ can be encoded
into a model of gradient nonholonomic flows.

We can formulate an evolution model for a family of geometric data
$\left(\overline{\mathbf{g}}(\tau ),\overline{\nabla }(\tau )\right) $ on
$\mathcal{T}^{E}\mathbf{E}$ induced by a family of regular
$L(\tau )\in C^{\infty }(E)$
with a flow parameter $\tau \in \lbrack -\epsilon ,\epsilon ]\subset \mathbb{%
R},$ when $\epsilon >0$ is taken sufficiently small. Let us introduce on the space of $Sec(E),$ for $\pi :E\rightarrow M,$ $\dim E=n+m$ and $\dim M=n\geq 2,$ the functionals
\begin{eqnarray}
\ _{\shortmid }\mathcal{F}(\overline{\mathbf{g}},\overline{\nabla },f,\tau)
&=&\int_{\overline{\mathcal{V}}}\left( \ _{\shortmid }R+\left| \overline{%
\nabla }f\right| ^{2}\right) e^{-f}\ dV,  \label{2pfrs} \\
\ _{\shortmid }\mathcal{W}(\overline{\mathbf{g}},\overline{\nabla },f,\tau)
&=&\int_{\overline{\mathcal{V}}}\left[ \tau \left( \ _{\shortmid }R+\left|
\overline{\nabla }f\right| \right) ^{2}+f-2m)\right] \mu \ dV,  \notag
\end{eqnarray}%
where the volume form $dV$ and scalar curvature $\ _{\shortmid }R$ are
determined by an off--diagonal metric $g_{\overline{\alpha }\overline{\beta }%
}$ (\ref{offd}). The integration is taken over $\overline{\mathcal{V}}%
\subset \mathcal{T}^{E}\mathbf{E},\dim \mathcal{V}=2m,$ corresponding to sections over a $U\subset M.$ We can fix $\int_{\overline{\mathcal{V}}}dV=1,$
with $\mu =\left( 4\pi \tau \right) ^{-m}e^{-f},$ considering necessary
classes of frame transforms and a parameter $\tau >0.$ The Ricci flow
evolution derived from (\ref{2pfrs}) in variables $\left( \overline{\mathbf{g%
}},\overline{\nabla }\right) $ is a standard theory for Riemann metrics \cite%
{ham1,ham2,gper1,gper2,gper3} but restricted to the conditions that such
metrics are induced by regular Lagrangians. The evolution in such variables is not adapted to a N--connection structure (\ref{nonalg}). It is possible
to elaborate N--adapted scenarios if above Perelman's functionals are
re--defined in terms of geometric data $(\widetilde{\mathbf{g}},\widehat{%
\mathcal{D}})$ and the derived flow equations are considered in N--adapted
variables.  Both approaches are equivalent if the distortion relations $%
\overline{\nabla }=\widehat{\mathcal{D}}+\widehat{\mathbf{Z}}$ (\ref%
{distrel1}) are considered for the same family of metrics, $\overline{%
\mathbf{g}}(\tau )=\widetilde{\mathbf{g}}(\tau )$ computed for the same set $%
L(\tau ).$

The theory of Lagrange--Ricci flows on $\mathcal{T}^{E}\mathbf{E}$ is
formulated as a model of evolving nonholonomic dynamical systems on the
space of equivalent geometric data $\left( L:\overline{\mathbf{g}},\overline{%
\nabla }\right) $ and/or $(L:\widetilde{\mathbf{g}},\widehat{\mathcal{D}})$
when the functionals $\ _{\shortmid }\mathcal{F}$ and $\ _{\shortmid }%
\mathcal{W}$ are postulated to be of Lyapunov type. Ricci flat
configurations (the Ricci tensor can be computed for one of the connections $%
\overline{\nabla }$ or $\widehat{\mathcal{D}})$ are defined as ''fixed'' on $%
\tau $ points of the corresponding dynamical systems.

We use $\breve{\tau}=\ ^{h}\tau =\ ^{v}\tau $ \ for a couple of possible $h$%
-- and $v$--flows parameters, $\breve{\tau}=(\ ^{h}\tau ,\ ^{v}\tau ),$ and
introduce a new function $\breve{f}$ instead of\ $\ f.$ The scalar functions
are re--defined in such a form that the ''sub--integral'' formula (\ref%
{2pfrs}) under the distortion of Ricci tensor (\ref{driccidist}) is
re--written in terms of geometric objects derived for the canonical
d--connection,
\begin{equation*}
(\ _{\shortmid }R+\left| \overline{\nabla }f\right| ^{2})e^{-f}=(\
_{s}^{F}R+|\ ^{F}\mathbf{D}\breve{f}|^{2})e^{-\breve{f}}\ +\Phi.
\end{equation*}
For the second functional, $\mathcal{D=(}h\mathcal{D},v\mathcal{D}),$ we
re--scale $\tau \rightarrow \breve{\tau}$ and write
\begin{equation*}
\left[ \tau (\ _{\shortmid }R+\left| \overline{\nabla }f\right| )^{2}+f-2m%
\right] \mu =[\breve{\tau}(~^{s}\widehat{\mathbf{R}}+|h\mathcal{D}\breve{f}%
|+|\ v\mathcal{D}\breve{f}|)^{2}+\breve{f}-2m]\breve{\mu}+\Phi _{1},
\end{equation*}%
for some $\Phi $ and $\Phi _{1}$ for which $\int_{\overline{\mathcal{V}}%
}\Phi dV=0$ and $\int_{\overline{\mathcal{V}}}\Phi _{1}dV=0.$ This provides
a proof for\footnote{%
similar N--adapted constructions were considered in Claim 3.1 in Ref. \cite%
{vricci1,vricci2} for nonholonomic manifolds and Lagrange spaces on $TM$}

\begin{lemma}
Considering distortion relations for scalar curvature and Ricci tensor \
determined by $\overline{\nabla }=\widehat{\mathcal{D}}+\widehat{\mathbf{Z}}$
(\ref{distrel1}), the Perelman's functionals (\ref{2pfrs}) are defined
equivalently in N--adapted variables $(L:\widetilde{\mathbf{g}},\widehat{%
\mathcal{D}}),$ {\small
\begin{eqnarray}
\mathcal{F}(\widetilde{\mathbf{g}},\widehat{\mathcal{D}},\breve{f}) &=&\int_{%
\overline{\mathcal{V}}}(~^{s}\widehat{\mathbf{R}}+|h\mathcal{D}\breve{f}%
|^{2}+|\ v\mathcal{D}\breve{f}|)^{2})e^{-\breve{f}}\ dV,  \label{2npf1} \\
\ \mathcal{W}(\widetilde{\mathbf{g}},\widehat{\mathcal{D}},\breve{f},\breve{%
\tau}) &=&\int_{\overline{\mathcal{V}}}[\breve{\tau}(~^{s}\widehat{\mathbf{R}%
}+|h\mathcal{D}\breve{f}|+|v\mathcal{D}\breve{f}|)^{2}+\breve{f}-2m]\breve{%
\mu}dV,  \label{2npf2}
\end{eqnarray}%
} where the new scaling function $\breve{f}$ satisfies $\int_{\overline{%
\mathcal{V}}}\breve{\mu}dV=1$ for $\breve{\mu}=\left( 4\pi \breve{\tau}%
\right) ^{-m}e^{-\breve{f}}$ and $\breve{\tau}>0.$
\end{lemma}

In this section, we omit details and proofs which are straightforward
consequences of those presented in \cite{gper1,gper2,gper3,caozhu}. For our
constructions, we consider operators defined by $L$ via $\overline{\nabla }$
on $\mathcal{T}^{E}\mathbf{E.}$ Using distortions to $\widehat{\mathcal{D}}$
with $\widehat{\mathbf{Z}}$ completely defined by $\widetilde{\mathbf{g}},$
we can study Lagrange--Ricci flows on prolongation Lie algebroids as
canonical nonholonomic deformations of Riemannian evolution on associated
vector/tangent bundles.

We can construct the canonical Laplacian operator, $\widehat{\Delta }:=$ $%
\widehat{\mathcal{D}}$ $\widehat{\mathcal{D}}$ determined by the canonical
d--connection $\widehat{\mathcal{D}},$ a ''standard'' Laplace operator $%
\overline{\Delta }=\overline{\nabla }\overline{\nabla },$ and consider
parameter $\tau (\chi ),$ $\partial \tau /\partial \chi =-1.$ For simplicity,
we shall not include the normalized term. The distortion (\ref%
{distrel1}) results in
\begin{eqnarray}
\Delta &=&\widehat{\Delta }+~^{Z}\widehat{\Delta },\ \ ^{Z}\widehat{\Delta }=%
\widehat{\mathbf{Z}}_{\overline{\alpha }}\widehat{\mathbf{Z}}^{\overline{%
\alpha }}+[\widehat{\mathbf{D}}_{\overline{\alpha }}(\ \widehat{\mathbf{Z}}^{%
\overline{\alpha }})+\widehat{\mathbf{Z}}_{\overline{\alpha }}\widehat{%
\mathbf{D}}^{\overline{\alpha }}];  \label{distb} \\
\ \overline{R}_{\ \overline{\beta }\overline{\gamma }} &=&\widehat{\mathbf{R}%
}_{\ \overline{\beta }\overline{\gamma }}+\widehat{\mathbf{Z}}ic_{\overline{%
\beta }\overline{\gamma }},\ \ _{s}R=\ _{s}\widehat{\mathbf{R}}+\widetilde{%
\mathbf{g}}^{\overline{\beta }\overline{\gamma }}\widehat{\mathbf{Z}}ic_{%
\overline{\beta }\overline{\gamma }}=\ _{s}\widehat{\mathbf{R}}+\ _{s}%
\widehat{\mathbf{Z}},  \notag \\
\ _{s}\widehat{\mathbf{Z}} &=&\mathbf{g}^{\overline{\beta }\overline{\gamma }%
}\ \widehat{\mathbf{Z}}ic_{\overline{\beta }\overline{\gamma }}=\ _{h}%
\widehat{Z}+\ _{v}\widehat{Z},\ _{h}\widehat{Z}=\widetilde{\mathbf{g}}%
^{\alpha \beta }\ \widehat{\mathbf{Z}}ic_{\alpha \beta },\ _{v}^{F}\widehat{Z%
}=\widetilde{\mathbf{g}}^{AB}\ \widehat{\mathbf{Z}}ic_{AB};  \notag \\
\ _{s}\overline{R} &=&\ _{h}\overline{R}+\ _{v}\overline{R},\ \ _{h}%
\overline{R}:=\widetilde{\mathbf{g}}^{\alpha \beta }\ \overline{R}_{\alpha
\beta },\ _{v}\overline{R}=\widetilde{\mathbf{g}}^{AB}\overline{R}_{AB},
\notag
\end{eqnarray}%
where, for convenience, capital indices $A,B,C...$ are for distinguishing $v$%
--components even the prolongation Lie algebroid is constructed for $\mathbf{%
P=E}.$ Using such deformations and a proof similar to that in Proposition
1.5.3 of \cite{caozhu}, we obtain

\begin{theorem}
\label{2theq1}The Lagrange--Ricci flows for $\widehat{\mathcal{D}}$ \
preserving a symmetric metric structure $\mathbf{\tilde{g}}$ and Lie
algebroid structure for prolongated $\mathcal{T}^{E}\mathbf{E}$ can be
characterized by this system of geometric flow equations:
\begin{eqnarray}
\frac{\partial \widetilde{\mathbf{g}}_{\alpha \beta }}{\partial \chi }
&=&-2\left( \widehat{\mathbf{R}}_{\alpha \beta \ }+\widehat{\mathbf{Z}}%
ic_{\alpha \beta }\right) ,\ \frac{\partial \widetilde{\mathbf{g}}_{AB}}{%
\partial \chi }=-2\left( \widehat{\mathbf{R}}_{AB}+\widehat{\mathbf{Z}}%
ic_{AB}\right) ,  \notag \\
\widehat{\ \mathbf{R}}_{\ \alpha A} &=&-\widehat{\mathbf{Z}}ic_{\alpha A},\
\ \widehat{\mathbf{R}}_{\ A\alpha }=\ \widehat{\mathbf{Z}}ic_{A\alpha },\
\label{frham1a} \\
\ \frac{\partial \widehat{f}}{\partial \chi } &=&-\left( \ \widehat{\Delta }%
+~^{Z}\widehat{\Delta }\right) \widehat{f}+\left| \left( \widehat{\mathbf{D}}%
+\widehat{\mathbf{Z}}\right) \widehat{f}\right| ^{2}-\ _{s}\widehat{\mathbf{R%
}}-\ _{s}\widehat{\mathbf{Z}},  \notag
\end{eqnarray}%
and the property that {\small
\begin{eqnarray*}
&&\frac{\partial }{\partial \chi }\mathcal{F}(\widetilde{\mathbf{g}},%
\widehat{\mathcal{D}},\widehat{f})=\int_{\overline{\mathcal{V}}}[|\widehat{%
\mathbf{R}}_{\alpha \beta \ }+\widehat{\mathbf{Z}}ic_{\alpha \beta }+(%
\widehat{\mathbf{D}}_{\alpha }+\widehat{\mathbf{Z}}_{\alpha })(\widehat{%
\mathbf{D}}_{\beta }+\widehat{\mathbf{Z}}_{\beta })\widehat{f}|^{2}+ \\
&&|\widehat{\mathbf{R}}_{AB\ }+\widehat{\mathbf{Z}}ic_{AB}+(\widehat{\mathbf{%
D}}_{A}+\widehat{\mathbf{Z}}_{A})(\widehat{\mathbf{D}}_{B}+\widehat{\mathbf{Z%
}}_{B})\widehat{f}|^{2}]e^{-\widehat{f}}dV,~\int_{\overline{\mathcal{V}}}e^{-%
\widehat{f}}dV=const.
\end{eqnarray*}%
}
\end{theorem}

\begin{proof}
For distortions (\ref{distb}), we can redefine the scaling functions from
above Lemma in different form. Similarly to \cite{vricci1,vricci2} we can
construct on $\mathcal{T}^{E}\mathbf{E}$ the corresponding system of Ricci
flow evolution equations for $\widehat{\mathcal{D}},\mathbf{\mathbf{\ }}$
\begin{eqnarray}
\frac{\partial \widetilde{\mathbf{g}}_{\alpha \beta }}{\partial \chi } &=&-2%
\widehat{\mathbf{R}}_{\alpha \beta \ },\frac{\partial \widetilde{\mathbf{g}}%
_{AB}}{\partial \chi }=-2\widehat{\mathbf{R}}_{AB},  \label{rfcandc} \\
\ \frac{\partial \widehat{f}}{\partial \chi } &=&-\widehat{\Delta }\widehat{f%
}+\left| \widehat{\mathcal{D}}\widehat{f}\right| ^{2}-\ _{h}\widehat{R}-\
_{v}\widehat{R},  \notag
\end{eqnarray}%
which can be derived from the functional $\widehat{\mathcal{F}}(\widetilde{\mathbf{g}%
},\widehat{\mathcal{D}},\widehat{f})=~\int_{\overline{\mathcal{V}}}(\ _{s}%
\widehat{R}+|\widehat{\mathbf{D}}\widehat{f}|^{2})$ $e^{-\widehat{f}}\ dV.$ The
conditions $\widehat{R}_{\alpha A}=0$ and $\widehat{R}_{A\alpha }=0$ must be
imposed in order to model evolution only with symmetric metrics. \ $\square $
\end{proof}

\vskip5pt

We note that under Ricci flows the N--adapted frames also depend on
parameter $\chi $ following certain evolution formulas. For $TM,$ such a
Corollary is proven in Ref. \cite{vricci2}. Re--defining indices for $%
\mathcal{T}^{E}\mathbf{E,}$ those formulas can be used for flow evolution of
frames of type (\ref{dderalg}) and (\ref{ddifalg}).

Finally, we discuss the statistical model which can be elaborated for Ricci
flows of mechanical systems. By definition, the functional $\ _{\shortmid }%
\mathcal{W}$ is analogous to minus entropy \cite{gper1} and this property
was proven for metric compatible nonholonomic and Lagrange--Finsler Ricci
flows \cite{vricci1,vricci2} with functionals $\widehat{\mathcal{W}}$ \
written for $\widehat{\mathbf{D}}.$ Similar constructions can be performed
on $\mathcal{T}^{E}\mathbf{E.}$

Let us consider a partition function $Z=\int \exp (-\beta E)d\omega (E)$ for the canonical ensemble at temperature $\beta ^{-1}$ being defined by the measure taken to be the density of states $\omega (E).$ The thermodynamical values are computed for average energy, $\ \left\langle E\right\rangle :=-\partial \log Z/\partial \beta ,$ entropy $S:=\beta \left\langle E\right\rangle +\log Z$ and fluctuation $\sigma :=\left\langle \left(E-\left\langle E\right\rangle \right) ^{2}\right\rangle =\partial ^{2}\log Z/\partial \beta ^{2}.$

\begin{theorem}
\label{theveq} Any family of Lagrangians under Ricci evolution on $\mathcal{T%
}^{E}\mathbf{E}$ is characterized by thermodynamic values {\small
\begin{eqnarray*}
\left\langle\tilde{E}\right\rangle &=&-\tilde{\tau}^{2} \int_{\overline{%
\mathcal{V}}}(\ _{s}\widehat{R}+|\widehat{\mathcal{D}}\tilde{f}|^{2}-\frac{m%
}{\widehat{\tau }}) \tilde{\mu}\ dV, \\
\tilde{S} &=& -\int_{\overline{\mathcal{V}}} [ \tilde{\tau}(\ _{s}\widehat{R}%
+|\widehat{\mathcal{D}}\tilde{f}|^{2}) +\tilde{f}-2m] \tilde{\mu}\ dV, \\
\tilde{\sigma} &=&2\ \tilde{\tau}^{4}~\int_{\overline{\mathcal{V}}}[|%
\widehat{\mathbf{R}}_{\overline{\alpha }\overline{\beta }}+\widehat{\mathcal{%
D}}_{\overline{\alpha }}\widehat{\mathcal{D}}_{\overline{\beta }}\tilde{f}-%
\frac{1}{2\tilde{\tau}}\mathbf{\tilde{g}}_{\overline{\alpha }\overline{\beta
}}|^{2}]\tilde{\mu}\ dV.
\end{eqnarray*}
}
\end{theorem}

\begin{proof}
Similar computations, in not N--adapted, or N--adapted forms, are given  in %
\cite{caozhu,vricci1,vricci2}. On prolongation Lie algebroids, we have to
use the partition function $\tilde{Z}=\exp \left\{ ~\int_{\overline{\mathcal{%
V}}}[-\tilde{f}+m]~\tilde{\mu}dV\right\} .$\ $\square $
\end{proof}

Finally, we note that this paper is a partner of \cite{vmedjm} and \cite{rf2016}, see \cite{alexiou1,alexiou2} on applications in modified gravity theories.

\section*{Acknowledgments}
 SV   research is related to his project activity at UAIC, a DAAD fellowship, and  the Program IDEI, PN-II-ID-PCE-2011-3-0256.

\end{document}